\documentclass[a4paper,UKenglish,cleveref, autoref,numberwithinsect]{lipics-v2019}

\title{A Simple Dynamization of Trapezoidal Point Location in Planar Subdivisions}

\author{Milutin Brankovic}{University of Sydney, Australia}{mbra7655@uni.sydney.edu.au}{}{}
\author{Nikola Grujic}{University of Sydney, Australia}{ngru0072@uni.sydney.edu.au}{}{}
\author{Andr\'e van Renssen}{University of Sydney, Australia}{andre.vanrenssen@sydney.edu.au}{https://orcid.org/0000-0002-9294-9947}{}
\author{Martin P. Seybold}{University of Sydney, Australia}{martin.seybold@sydney.edu.au}{https://orcid.org/0000-0001-6901-3035}{}

\authorrunning{M. Brankovic, N. Grujic, A. van Renssen, M. P. Seybold}

\Copyright{Milutin Brankovic, Nikola Grujic, Andr\'e van Renssen, Martin P. Seybold}



\supplement{\url{https://github.com/milutinB/dynamic_trapezoidal_map_impl}}

\acknowledgements{We want to thank the University's undergraduate Research Internship Program, which fostered the pursuit of this topic.}

\usepackage[utf8]{inputenc}

\usepackage{amsmath,amsthm,amssymb}
\usepackage{xcolor, hyperref}
\hypersetup{
    colorlinks,
    linkcolor={red!90!black},
    citecolor={green!80!black},
    urlcolor={blue!80!black}
}
\usepackage{graphicx}
\usepackage{algorithm}

\usepackage[disable]{todonotes}

\hideLIPIcs  
\nolinenumbers 

\newcommand{\R}{\mathbb{R}}
\newcommand{\Q}{\mathbb{Q}}
\renewcommand{\O}{\mathcal{O}}

\newcommand{\A}{\mathcal{A}}
\newcommand{\T}{\mathcal{T}}
\newcommand{\D}{\mathcal{D}}

\newtheorem{obs}{Observation}

\newcommand{\avg}{\displaystyle \mathop{\mathbb{E}}}
\newcommand{\Perm}{\mathbf{P}}

\newcommand{\partition}{\operatorname{Partition}}
\newcommand{\merge}{\operatorname{Merge}}
\newcommand{\vPart}{\operatorname{V-Partition}}
\newcommand{\vMerge}{\operatorname{V-Merge}}

\newcommand{\leafIns}{\operatorname{Leaf-Insert}}

\newcommand{\subtrees}{\operatorname{Descend}}

\begin{document}
\maketitle
\begin{abstract} 
	We study how to dynamize the Trapezoidal Search Tree (TST) -- a well known randomized point location structure for planar subdivisions of kinetic line segments.

	Our approach naturally extends incremental leaf-level insertions to recursive methods and allows adaptation for the online setting.
	Moreover, the dynamization carries over to the Trapezoidal Search DAG (TSD), offering a linear sized data structure with logarithmic point location costs as a by-product.
	On a set $S$ of non-crossing segments, each update performs expected $\O(\log^2|S|)$ operations.

	We demonstrate the practicality of our method with an open-source implementation, based on the Computational Geometry Algorithms Library, and experiments on the update performance.
\end{abstract}

\section{Motivation}
Our search for a simple data structure, with guarantees, for practical use is motivated by applications that seek to built the intersection graph of thin geometric objects.
Already $n$ line segments in the plane can have $k=\Theta(n^2)$ intersections (e.g. grid placement) and a naïve algorithm, checking all pairs of objects for intersections, has optimal worst-case time complexity with very little need for additional memory space.
We focus on reporting problems for instances with fewer, $k \in o(n^2)$, intersections.
%
Creating intersection graphs is a natural precursor in meshing fracture networks~\cite{Hobe18}.
Simple examples of such networks are sets of plane, simple polygons in $\R^3$ with bounded complexity (e.g. rectangles). 
The classic reporting problems for axis orthogonal rectangles are well understood and solvable with linear sized structures with $\O(n\log n+k)$ operations for $\R^2$~\cite{BentleyW80} and with $\O(n\log^2 n+k)$ operations for $\R^3$~\cite{ChazelleI84, EdelsbrunnerO85}.
Rectangles with orientation however seem to be out of reach for these search tree approaches.
Sweeping-plane approaches reduce the intersection tests for a newly encountered object to those within a list of currently active objects.
A naïve-sweep succeeds with $\O(n \log n + n\cdot m)$ operations, for $m \leq n$ bounding the size of any active list occurring in the sweeping direction.
However, for a $m \in \omega(\sqrt{n})$ many futile intersection tests may be performed.

This note studies how to structure line segment data in such sweeping planes.
That is, we seek to provide a fully dynamic data structure for kinetic line segments in a plane that supports intersection reporting by means of ray-shooting.
Such data structures are also of interest for map matching applications
\cite{Seybold17} or for applications in need of a point location structure to maintain a constrained Delaunay triangulation~\cite{Bahrdt17}.
%

\subsection{Related Work}
Planar subdivisions are often categorized by additional geometric properties on the subdivisions' faces.
Authors frequently consider cases where all faces are convex, or all faces intersect horizontal lines at most twice (monotone faces).
Clearly, monotone faces and faces' boundaries formed by monotone curves (e.g. line segments) are distinct properties.
In particular, general planar subdivisions of graphs with straight-line embedding have monotone face boundaries.

Trapezoids are a key ingredient to simplify the study of geometric problems concerning sets of line segments $S$.
We use three well known \cite{AgarwalEG98, BergCompGeo, Mulmuley90, Seidel91} and closely related \cite{Hemmer16} structures.
These are
(i)   the trapezoidal decomposition $\A(S)$ of a plane induced by $S$, 
(ii)  the Trapezoidal Search Tree (TST), and 
(iii) the Trapezoidal Search DAG (TSD)
(c.f. Section~\ref{sec:defs}).
TSTs $\T(S, \pi)$ and TSDs $\D(S,\pi)$ stem from deterministic algorithms, that incrementally insert the segments in $S$ according to a, typically random, permutation $\pi$ over $S$.
Decomposition $\A(S)$ has size $\O(|S|+k)$, where $k$ denotes the number of \emph{crossing pairs} of segments in $S$.
Since segment boundaries of a general planar subdivision are non-crossing ($k=0$), the refined decomposition $\A(S)$ of the boundary segments $S$ is within a constant of the subdivision's size, regardless of the face shape.

Chiang and Tamassia's 1992 survey~\cite[Chapter~6]{Tamassia92survey} reviews several dynamizations for planar subdivisions.
In~\cite{ChiangT92}, the two authors describe a dynamization of the trapezoidal decomposition $\A$ of monotone subdivisions.
Using several tree structures, they achieve $\O(\log |S|)$ point location query time within an $\O(|S| \log |S|)$ size structure, which allows fully-dynamic updates in amortized $\O(\log^2 |S|)$ operations.
Also for monotone subdivisions, Goodrich and Tamassia show with~\cite{GoodrichT98} how to maintain interlaced, monotone spanning trees over the edges of the planar graph and its dual graph.
This leads to an $\O(|S|)$ sized structure with an update time of $O(\log |S| + m)$, for insertion of a monotone chain of length $m$, but point location queries take $\O(\log^2 |S|)$ operations.

Most recent approaches that achieve \emph{online, fully-dynamic updates} within $\O(|S|)$ size and $\O(\log|S|)$ query time are extensions of the works of Baumgarten et al.~\cite{BaumgartenJM94} based on dynamic fractional cascading.
Arge et al.~\cite{ArgeBG06} showed that raising the fan-out of tree nodes to $\Theta(\log^\varepsilon |S|)$ allows to lower the point location query time bound to $\O(\log |S|)$.
However, insertions take $\O(\log^{1+\varepsilon}|S|)$ amortized time and deletions take $\O(\log^{2+\varepsilon} |S|)$ amortized time per update.
Chan and Nekrich~\cite{ChanN15} added Multi-Colored Segment Trees to the approach and applied de-randomization and de-amortization techniques to finally derive update bounds of $\O(\log^{1+\varepsilon}|S|)$ for both, insertions and deletions.
To our knowledge, these are currently the best bounds on the pointer machine.
Refining this approach, Munro and Nekrich~\cite{MunroN19} recently analyzed the I/O counts in an external memory model with sufficiently large block sizes of $\Omega(\log^8 n)$.

The \emph{Randomized Incremental Construction (RIC)} of TSDs provides simple, and thereby practical, point location structures of expected $\O(|S|)$ size whose longest search path is with high probability also in $\O(\log|S|)$.
Their analysis is a celebrated result of Mulmuley~\cite{Mulmuley90}, who argues on a random experiment that draws geometric objects, and Seidel~\cite{Seidel91}, who uses a backward-analysis to bound the expected costs.
Vastly unmentioned is Mulmuley's book~\cite{MulmuleyBook}, which extends TSDs to a specialized (offline) dynamic setting of a random insert/delete string for a fixed set of segments $S$, in which each element has a probability of $1/|S|$ to appear in each update request.
The approach thrives on a randomly chosen, fixed order on $S$, rotations in the structure, and regular leaf-level insertions and deletions.
Several distributed book chapters %
finally show that executing such a (+/-)-sequence of $m$ updates has a total expected cost of $\O(m \log m + k \log m)$, where $k$ is the number of intersections in the fixed set $S$.
Schwarzkopf~\cite{Schwarzkopf91} independently describes a similar update model that also leads to a dynamization.
His analysis provides an expected $\O(\log^2 n)$ time for an update operation from the sequence, where $n$ denotes the current number of non-crossing segments, but the expected point location query bound is w.h.p. in $\O(\log^2 n)$.

In~\cite{AgarwalEG98}, Agarwal et al. describe a TST over a static set of segments, that move over time.
Their approach is based on a RIC of TSTs with expected $\O(|S|\log|S|+k)$ size whose depth is w.h.p in $\O(\log|S|)$.
The authors mention a bound on the expected construction time of $\O(|S|\log^2|S| +k\log|S|)$ and show that randomization is crucial to resolve a structural change, due to kinetic (adversarial) movement of the line segments, in expected $\O(\log|S|)$ operations.
One may however form constant velocity trajectories for (point-like) segments that lead to $\Omega( n\sqrt{n} )$ many topological updates for any binary space partition method~\cite{AgarwalBBGH00}.

Insightful work of Hemmer et al.~\cite{Hemmer16} recently revealed a certain bijection between the search paths of $\T(S,\pi)$ and the `un-folded' search paths of $\D(S,\pi)$ by means of a structural induction argument (along fixed $\pi$) on the two incremental algorithms.
The authors show Las Vegas type point location guarantees for the RIC of TSDs by bounding computation time of the length of a longest search path in a TSD by means of \emph{the size} of the related TST.
	
%

\subsection{Contribution}
Our surprisingly simple approach uses only one search structure and performs insertions and deletions directly on higher levels of it.
Since we maintain randomness of the priority orders, the expected size and point location query bounds are retained.

Section~\ref{sec:defs} reviews the basic RIC of TSTs and TSDs in a unified notation and we introduce our natural, recursive extensions of those basic primitives in Section~\ref{sec:algos}.
This enables fully-dynamic updates in both data structures in the offline setting, which we extend using Treaps to support efficient priority comparisons for the online setting (c.f. Section~\ref{sec:online-vs-offlie}).
Our algorithms are suitable for pointer machines (with arithmetic) since they only compare the spatial location of the input points, intersection points, and integer priority values. 
In particular, the method does not require indirect addressing of Random Access Memory.

Our analysis uses a geometric interpretation of TST nodes and their induced refinement of the neighborhood zones in certain trapezoidal decompositions $\A$.
Based on this, we perform a backward analysis to determine the cost of emplacing elements in a priority order.
This also allows us to bound the expected cost of simpler `segment' searches for affected search nodes in the structure.
For TSTs and TSDs, we obtain an expected update costs of $\O(\log^2 |S|)$ for insertion and deletion.
Hence randomization allows improvements on the amortized deletion bound of \cite{ArgeBG06}, without losing the size or point location query bounds.
For non-crossing segments, our analysis of dynamic TST updates matches the RIC asymptotic of $\O(|S|\log^2|S|)$.

We provide an open-source implementation for the dynamic TST based on the Computational Geometry Algorithms Library~\cite{CGAL} and present experimental data (c.f. Section~\ref{sec:exp}).
The observed update performance on non-crossing segments matches our analysis tightly and suggests small asymptotic constants.

\section{Basic Definitions, Algorithms and Properties} \label{sec:defs}
A segment $s=\overline{pq}$ between two distinct points 
$p,q \in \R^2$ is the set
$\{ p + \alpha (q-p) :\alpha \in [0,1]\}$ of points in the Euclidean plane.
Two segments are called disjoint if $s \cap s' = \emptyset$ and \emph{overlapping} if $s\cap s'$ contains more than one point.
For the case of one common point, we say that two segments \emph{meet} if it is an endpoint that is contained in the other segment and otherwise they are called \emph{intersecting} (or \emph{crossing}). 
The segment boundaries of a general planar subdivision for example have no intersections, but may well contain points in which arbitrarily many segments meet.

A line, parallel to the $y$-axis, through a point is called a \emph{vertical-cut} (e.g. through an end or intersection point) and the line through the endpoints of a segment is called an \emph{edge-cut}.
To avoid ambiguity, we denote the open halfspaces of a cut with $-$ and $+$.
For vertical-cuts $-$ denotes the one with lower $x$-coordinates and for edge-cuts $-$ denotes the one with the lower values of the $y$-axis.

As in \cite[Chapter 6]{BergCompGeo}, our presentation assumes that no two distinct points (end or intersection) have exactly the same $x$-coordinate, unless they are a common endpoint in which the segments meet.
As usual, this assumption is resolved by an implicit, infinitesimal shear transformation which translates to a tie-breaking rule in the geometric orientation predicate that resolves to comparing the  $y$-coordinates.
A simple extension of the halfspace partitioning for the remaining cases (of segments and cuts due to a segment) conceptually moves the points infinitesimally away from the degeneracy along their segment for a consistent decision.
To also assign co-linear cases consistently we finally resolve to lexicographic comparison (c.f. Figure~\ref{fig-tiebreaking}).
\begin{figure}
	\includegraphics[width=\linewidth]{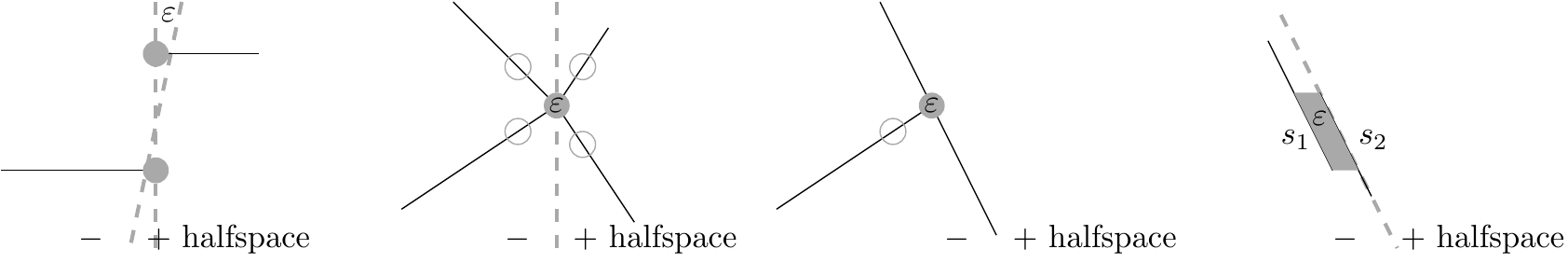}
	\caption{Standard shear transformation (left) and three refined tie-breaking rules for point/point meets, point/segment meets (geometric), and segment/segment overlaps (lexicographic).} \label{fig-tiebreaking}
\end{figure}
Moreover, our presentation assumes that the whole domain is bounded by a very large rectangle, which is also a trapezoid.
Extensions to unbounded problems are commonly achieved by careful case treatment of unbounded faces of arrangements (e.g. in \cite{CGAL}).

Our combinatorial analysis (c.f. Section~\ref{sec:runtime}) of the proposed dynamization of the TST requires more precise terms to establish a geometric correspondence, which we introduce now in a unified review of the well known incremental algorithms.
We frequently identify a permutation $\pi \in \Perm(S)$ over $S$ as bijection $\pi: S \to \{1,\ldots, |S|\}$ and call
$\pi(s)$ the \emph{priority} of a segment $s \in S$.

\subsection{Trapezoidal Search Trees $\T(S,\pi)$}\label{sec:def-tst}
A TST over a trapezoidal domain $\Delta \subseteq \R^2$ is a certain hierarchical Binary Space Partition with (vertical and edge) cuts induced by the segments in $S$.
These rooted trees are defined by inductively applying the deterministic insertion algorithm for the segments of $S$ in \emph{ascending priority} order.
Hence, the structure is uniquely determined by $(S,\pi)$.

Every node $v$ of the binary tree is associated with a trapezoidal region $\Delta(v) \subseteq \Delta $ and, on the empty set of segments, $\T(\{\},())$ contains only the root $r$ with $\Delta(r)=\Delta$.
In TSTs, every non-leaf node $v$ has two child nodes $v^-, v^+$ and stores a cut that signifies the halfspace partition of $\Delta(v)$ in $\Delta(v^-)$ and $\Delta(v^+)$, which enables point-location search descends.
We say that this cut \emph{destroys} the trapezoidal region $\Delta(v)$ and use this to extend the priority assignment to comprise also tree nodes.
We denote the priority of a node $v$ with $p(v)$.
More precisely, leafs have priority $+\infty$ and non-leaf nodes have the priority of the segment whose cut insertion destroys the node.
E.g. we have for tree nodes that the priorities of the cuts that constitute the ($3$ or $4$) boundaries of $\Delta(v)$ are always smaller than $p(v)$.

See Algorithm~\ref{alg:l-insert-TST} for the steps that are performed on leafs to insert a new segment $s$ in the TST structure $\T$.

\begin{algorithm}[h!]
\caption{$\leafIns(\T, s)$:} \label{alg:l-insert-TST}
\begin{enumerate}
    \item Search for the leaf nodes $L=\{u_1, \ldots, u_l\} \subseteq \T$ with $\Delta(u_i) \cap s \neq \emptyset$.
    \item Create refined slabs of these regions by vertically partitioning each $\Delta(u_i)$ with cuts due to endpoints of $s$ or intersection points of $s$ with edge-cuts bounding $\Delta(u_i)$.
    \item Partition the intersected slab regions further with the edge cut through $s$.
\end{enumerate}
\end{algorithm}

Note that nodes' priorities are monotonically increasing on paths from the root.
Moreover, modifications of $\leafIns$ on a leaf are independent from those on other leafs.
Procedure $\leafIns$ inserts $1,2$ or $3$ cuts on the region of a leaf and the edge-cut is always the last cut that is performed.
We refer to these patterns as vertical-vertical-edge (VVE), vertical-edge (VE), and edge (E) destruction.
To remove unnecessary ambiguity in the tree structure, we use the additional convention that in a VVE-destruction, the left of the two vertical cuts is inserted first (e.g. as parent of the right vertical cut).
See Figure~\ref{fig:example-leaf-insert-tst} for an example.

\begin{figure}\centering
	\includegraphics[width=.3\linewidth]{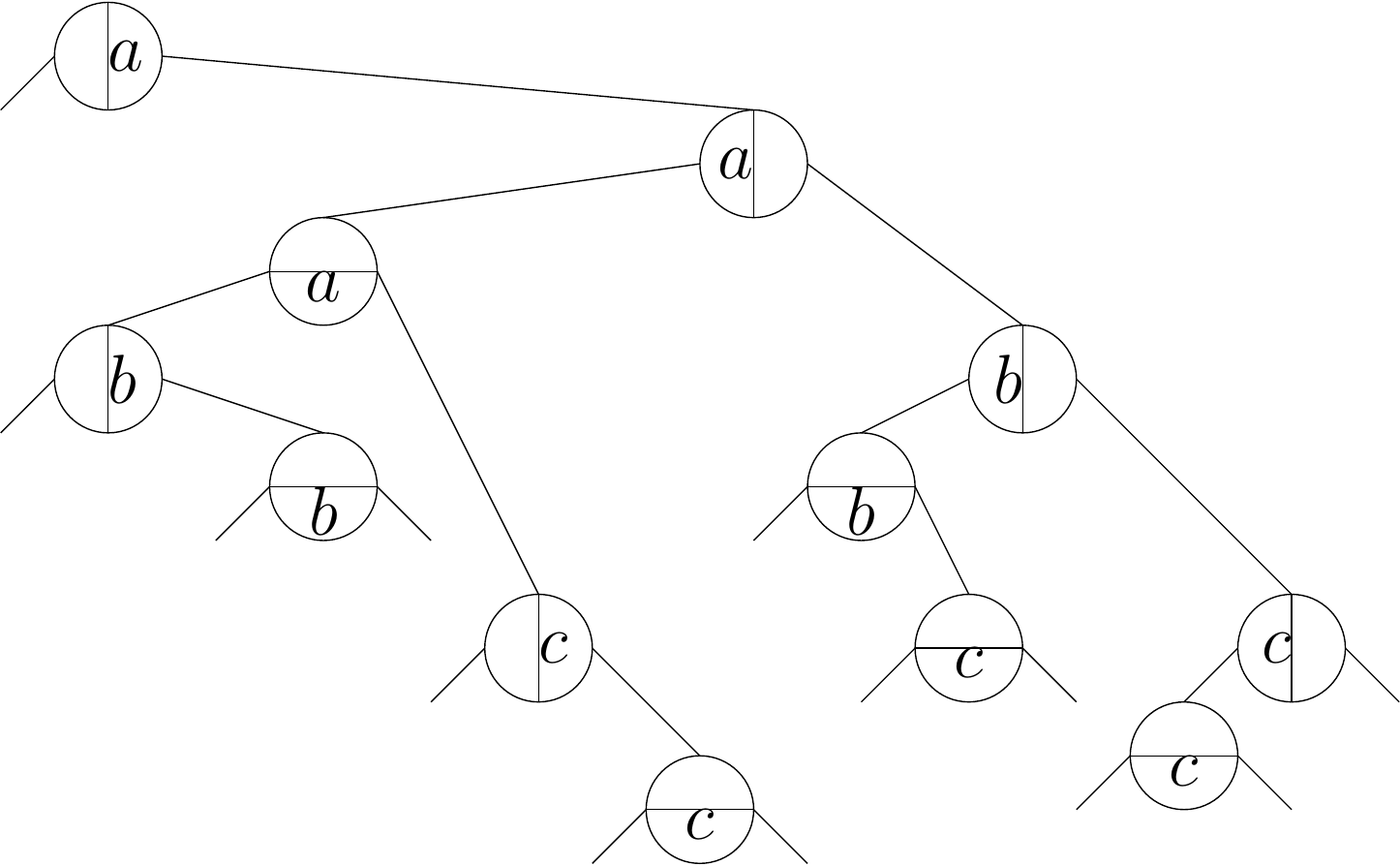}~
	\includegraphics[width=.35\linewidth]{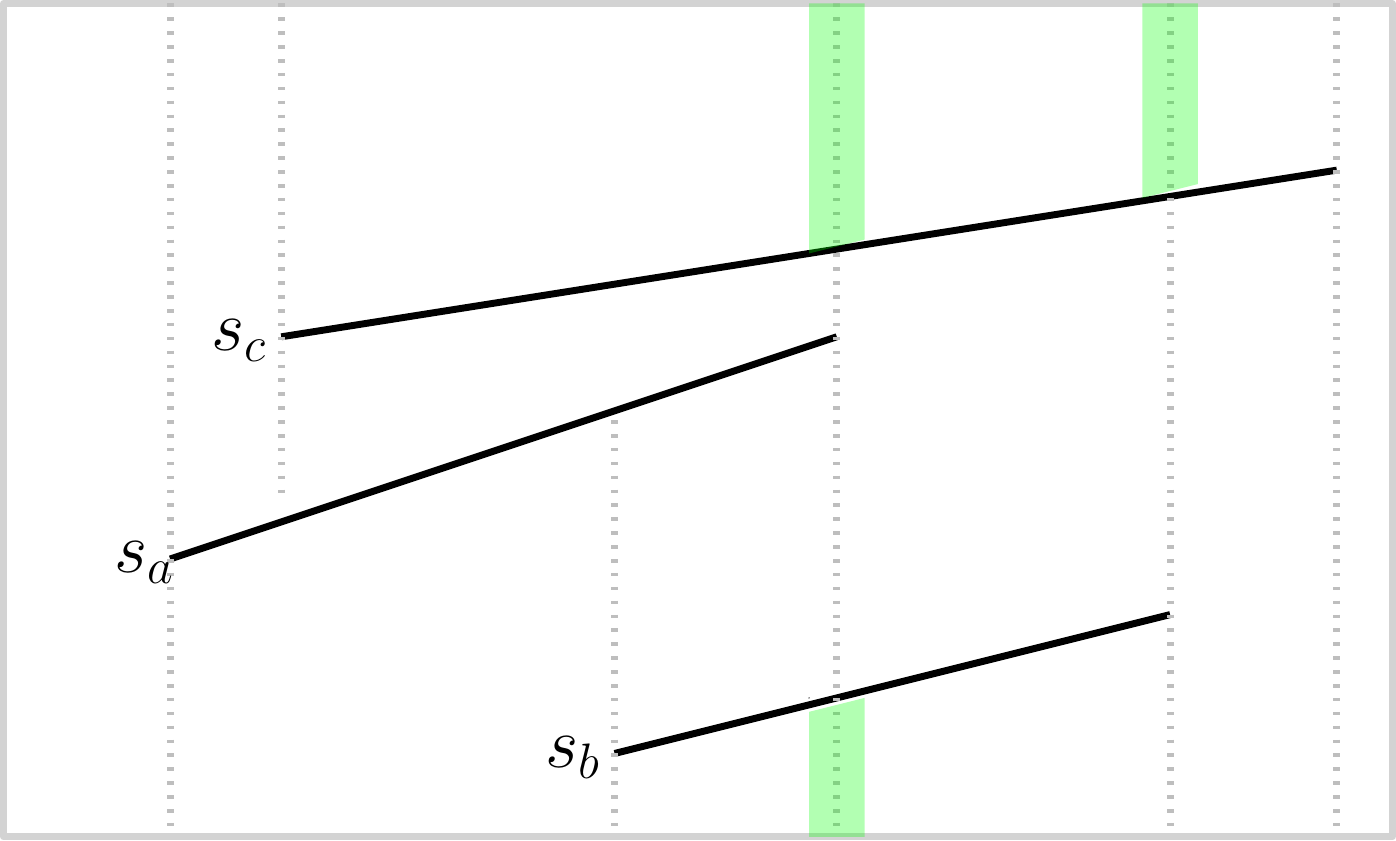}~
	\includegraphics[width=.3\linewidth]{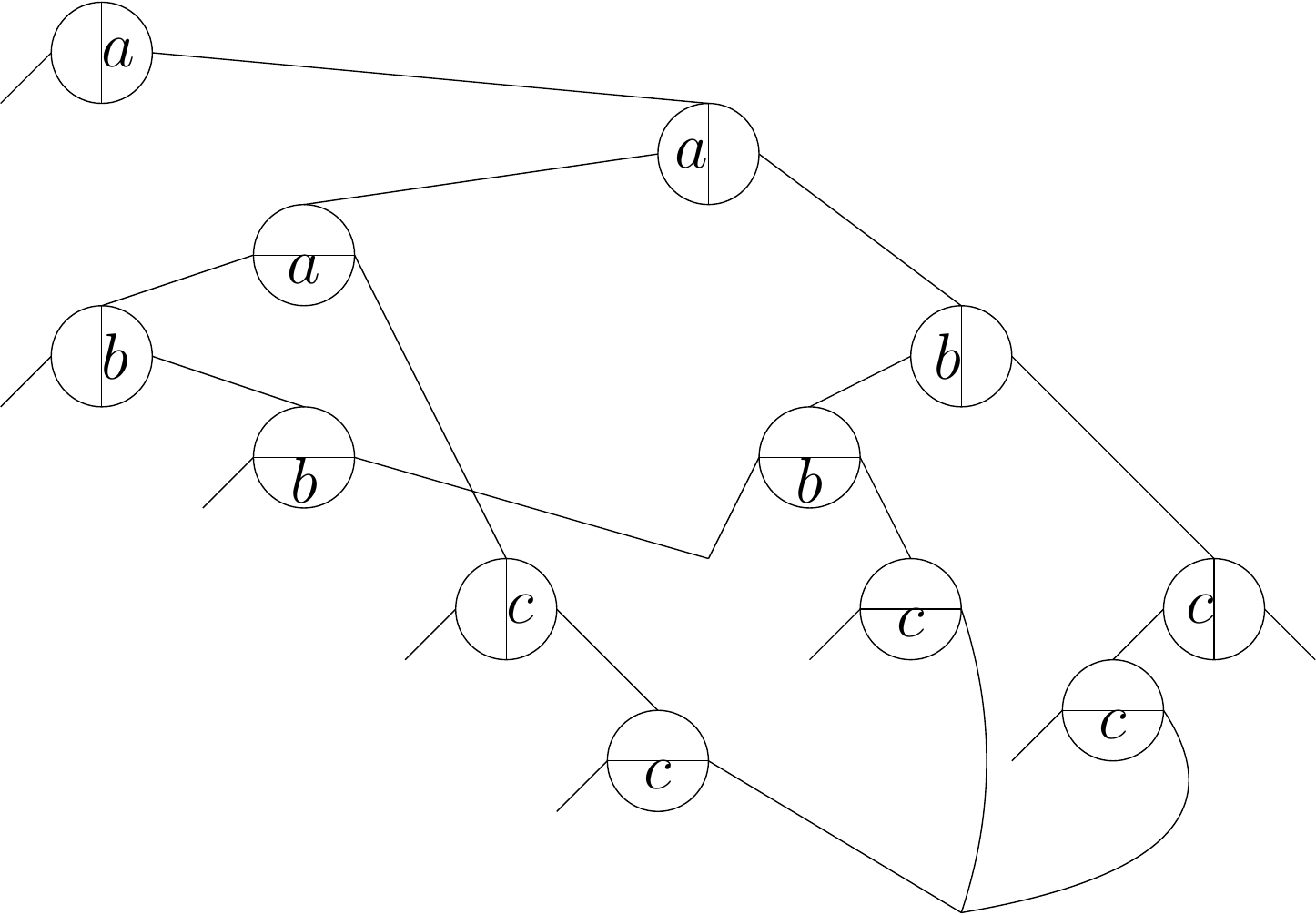}
	\caption{Example of TST $\leafIns$ (left) and TSD $\leafIns$ (right) for segments with priority order $s_a < s_b < s_c$ (middle). The vertical merges of Algorithm~\ref{alg:l-insert-DAG} are indicated in green.}\label{fig:example-leaf-insert-tst}
\end{figure}

\begin{theorem}[TST size and depth~\cite{AgarwalEG98}]\label{thm:exp-tree-size}
Let $S$ be a set of segments in $\R^2$ and $k_S$ the number of intersecting pairs among them.
The expected size of a TST over $S$ is bounded by
$$\avg_{\pi \in \Perm(S)}|\T(S,\pi)|  = \O(|S|\log |S| + k_S) ~.$$
The expected depth of leafs is $\O(\log |S|)$ and the maximal leaf depth of $\T$ is w.h.p. $\O(\log |S|)$.
\end{theorem}

Clearly, any TST over a set of segments $S$ contains at least $\Omega(|S| + k_S)$ nodes and there are certain instances that have a worst-case size of $\Omega(|S|^3)$ -- E.g. unfortunate insertion orders of segments with $k_S = \Omega(|S|^2)$.

Agarwal et al. \cite{AgarwalEG98} also mention an upper bound on the expected, purely incremental construction time of $\O(|S|\log^2 |S| + k_S \log|S|)$, which we match with the analysis of our dynamic updates (c.f. Section~\ref{sec:runtime}).

\subsection{Trapezoidal Decomposition $\A(S)$ and Search DAG $\D(S,\pi)$}\label{sec:def-tsd}
One may save space in such a search structure by considering a certain planar subdivision $\A(S)$ that is induced from a set of segments $S$.
Given aforementioned discussion of degeneracies, this subdivision is defined by emitting two vertical rays (in negative and positive $y$-direction) from each end and intersection point until the ray meets the first segment or the bounding rectangle.
Each face of $\A(S)$ is a trapezoid with $3$ or $4$ boundary edges and standard double counting establishes a size of $\O(|S|+k_S)$ for these planar decompositions.
Note that $\A(S)$ is independent of segment priorities.

Mulmuley~\cite{Mulmuley90} and Seidel~\cite{Seidel91} consider the incremental process of inserting an additional segment $s$ and obtaining $\A(S \cup \{s\})$ from $\A(S)$.
Algorithm~\ref{alg:l-insert-DAG} shows the additional merging phase, to `contract' pre-existing vertical-cuts (between the regions of the nodes in $L$) towards their emission point until they meet $s$, to maintain the property that the leaf regions coincide with $\A(S)$.
Hence, the planar subdivision of the leafs of a TST $\{\Delta(v) \subseteq \R^2:v \in \T(S,\pi), v \text{ is leaf}\}$ is a \emph{refinement} of $\A(S)$, regardless of $\pi$.

\begin{algorithm}[h!]
	\caption{$\leafIns(\D, s)$:} \label{alg:l-insert-DAG}
	\begin{enumerate}
		\item Search for the leaf nodes $L=\{u_1, \ldots, u_l\} \subseteq \D$ with $\Delta(u_i) \cap s \neq \emptyset$. \label{alg-line:l-insert-DAG-search}
		\item Create refined slabs of these regions by vertically partitioning each $\Delta(u_i)$ with cuts due to endpoints of $s$ or the intersection points of $s$ with edge-cuts bounding $\Delta(u_i)$.
		\item Partition the intersected slab regions further with the edge cut through $s$.
		\item Reconsider the vertical cuts between nodes in $L$ that cross $s$.\\
		Merge the vertical cuts of the refined regions whose emission point is now blocked by $s$.
	\end{enumerate}
\end{algorithm}

Though the resulting search graph is a DAG (also known as History DAG), the basic data stored with a node is still identical with the binary TST tree nodes (e.g. Figure~\ref{fig:example-leaf-insert-tst}).
One key ingredient for the well known, expected construction time of ${\O(|S|\log|S| + k_S)}$, is to perform the search for affected leafs (c.f. Step~\ref{alg-line:l-insert-DAG-search}) quickly.
A way to achieve this, without prior knowledge of all future segment insertions, is to also maintain an explicit graph representation of the planar subdivision $\A(S)$ of the segments $S$ that are currently contained in the structure.
This trick allows to perform only one point location query for, e.g., the left endpoint of $s$, followed by a walk along $s$ through $\A(S)$.
We summarize the well known aspects of this incremental algorithm in the following statement.

\begin{theorem}[TSD size and depth~\cite{BergCompGeo, Mulmuley90, Seidel91}]\label{thm:exp-dag-size}
	Let $S$ be a set of segments in $\R^2$ and $k_S$ the number of intersecting pairs among them.
	The expected size of a TSD over $S$ is bounded by
	$$\avg_{\pi \in \Perm(S)}|\D(S,\pi)|  = \O(|S| + k_S) ~.$$
The expected search path length is $\O(\log |S|)$ and the maximal search path length of $\D$ is w.h.p. in $\O(\log |S|)$.
\end{theorem}

The TST and TSD structure also allow simple deletions of segments in a certain decremental setting, that deletes all segments in \emph{descending} priority order.
That is, the modifications that the structure undergoes in the steps of Algorithm~\ref{alg:l-insert-DAG} are simply undone in exactly the reverse order (of all $4$ steps), causing the same amount of work.

The insightful work of \cite{Hemmer16}, recently revealed a certain bijection between the search paths of $\T(S,\pi)$ and the (valid) search paths in $\D(S,\pi)$ by means of a structural induction argument (along fixed $\pi$) on the two incremental algorithms.
This allows them to bound the runtime of computing the length of a longest search path in a TSD $\D(S,\pi)$ by means of the size of the related TST $\T(S,\pi)$.
Based on this technique, our analysis of the update time in TSTs (c.f. Section~\ref{sec:runtime}) carries over to an upper bound in TSDs as well.

The main geometric difference between the structures nodes' regions is, that in TSTs the region of a node is always a subset of its parents' region, whereas in TSDs the region of a node may extend further to the left and right, due to vertical merges, but not across the top and bottom boundaries.
We now describe our recursive algorithms for vertical partitions and merges, and (non-vertical) partitions and merges that operate on intermediary priority levels (e.g. higher up in the search structures).

%
%
\section{Recursive Primitives for Dynamic Updates}\label{sec:algos}
Characteristic for our approach is that insertions and deletions are performed directly on higher levels of the search structures rather than solely on leafs.
We first describe the recursive primitives for inserting a new segment $s$ in the structure of $\T(S,\pi)$.

We choose a random position $p \in \{1,\ldots, |S|+1\}$ uniformly in which we emplace the element $s$, between $(p-1)$ and $p$, in the sequence $\pi$, calling the resulting priority order $\pi'$.
Our recursive algorithms then update the structure $\T(S,\pi)$ exactly to $\T(S \cup \{s\},\pi')$.
More precisely, we restrict the search (c.f. Step~\ref{alg-line:l-insert-DAG-search} of Algorithms~\ref{alg:l-insert-TST} and \ref{alg:l-insert-DAG}) to those nodes with priority at most $p$, which leads to a set $L$ of affected subtree roots whose priorities are larger than $p$ in $\pi'$.
We conceptually `hang out' these nodes by creating a copy $L'$ of them and reverting those in $L$ to leafs, temporarily.
The insertion then proceeds in exactly the same sequence as on regular leafs, but every binary space partition on nodes of $L$ is accompanied by the matching recursive call on the subtree of the node's copy in $L'$.
After this process, we have a matching subtree root in $L'$ for each leaf that was created below $L$, which we then `hang in' instead of the simple leaf.
Given the close relation of the two structures, it turns out that our recursive primitives, with few changes, already provide the necessities for dynamic updates of TSDs.

To simplify presentation in this and the following section, we assume that the priority values, stored with the segments, are integers in $\{1, \ldots, |S|+1\}$ that were already updated to represent $\pi'$.
Therefore $\pi'(s') > \pi'(s)$ comparisons are decidable in constant time.
Our Section~\ref{sec:online-vs-offlie} shows a simple solution, based on Treaps, to resolve this strong assumption sufficiently fast for the online setting.

Given the discussion above on the possible E-,VE- and VVE-destruction patterns, we simplify notation in this section by denoting with $\subtrees(v)$ the set of $0, 2, 3$, or $4$ descendants with the next higher priority value than $p(v)$.
E.g. $\subtrees(v)=()$ if $v$ is a leaf, $\subtrees(v)=(v_a,v_b)$ if $v$ underwent an E-destruction, $\subtrees(v)=(v_l,v_a,v_b)$ or $\subtrees(v)=(v_a,v_b,v_r)$ if $v$ underwent a VE-destruction, and $\subtrees(v)=(v_l,v_a,v_b,v_r)$ if $v$ underwent a VVE-destruction (e.g. Figure~\ref{fig:vPart-vMerge-Tree}).

\subsection{Priority Restricted Searches} \label{sec:segment-search}
Locating the affected subtree roots for our approach seems challenging without an explicit representation of $\A$ for the segments with priority at most $p$.
Given the monotonically increasing priorities on search paths,
we use a simple iterated ray shooting walk to find all nodes $L=\{u_i\}$, with minimal $p(u_i)>p(s)$, whose regions $\Delta(u_i)$ cover $s$.

To show sufficiently low expected bounds in Lemma~\ref{lem:segment-search},
we find the following alternative description helpful.
This top-down refinement process also derives the node set $L$, such that its nodes $u_i$ are sorted by the sequence in which $s$ stabs $\Delta(u_i)$ from left-to-right:
First locate the search node $u$, with maximal $p(u) < p(s)$, that fully contains $s$ in $\Delta(u)$.
Place $u$ in an initially empty list.
Successively replace the leftmost node $u$ in the list with $p(u) < p(s)$ with the nodes of $\subtrees(u)$, whose region intersects $s$, until all node priorities exceed $p(s)$.

Note that nodes with priority larger than $p(s)$ are not refined and the spatial location of the regions of the nodes in $\subtrees(\cdot)$ allows us to easily keep the set $L$ sorted by left-to-right stabbing sequence of $s$.

\subsection{Recursive Vertical Partitions and Merges}
Our methods are inspired by sorting algorithms on lists of integers, though they move cuts and child relations among nodes based on the nodes' priorities.
We first introduce the recursive primitive $\vPart(u, q,v^-,v^+)$, where $q \in \R^2$ denotes the point that induces the vertical cut $c(q)$ and $u$ is an affected tree node.
E.g. $p(u)>p(s)$ but the boundary cuts of $\Delta(u)$ have priority lower than $p(s)$.
As outlined above, the primitives update the search structure to the state that regular leaf insertion (in ascending priority order) would have created under the presence of vertical cut $c(q)$ splitting $\Delta(u)$ in $\Delta^- \cup \Delta^+$.
(A lengthy, rigorous proof fixes $\Delta(u)$ and performs a structural induction argument over the sequence of successive cuts that destroy the region.)
See Figure~\ref{fig:vPart-vMerge-Tree} for the recursive cases by destruction patterns (c.f. Algorithms~\ref{alg:v-part} and \ref{alg:v-merge} in Appendix~\ref{sec:code}).

The recursive node visits are similar to a search for points on $c(q)$ and $\vPart$ performs at most two recursive calls per node in TSTs.
Reverting these steps in exactly the reverse order provides the inverse operation $\vMerge(u^-,u^+, q, v)$.
Note that the TSD primitives actually perform \emph{fewer} recursive calls, since vertical cuts are contracted in this structure.
However, a small number of additional `bouncing' nodes (c.f. Section $4$ in \cite{Hemmer16}) may be visited before the next strict refinement of the search region happens (c.f. Section~\ref{sec:runtime}).

\begin{figure}[]
	\centering
	\includegraphics[width=\linewidth]{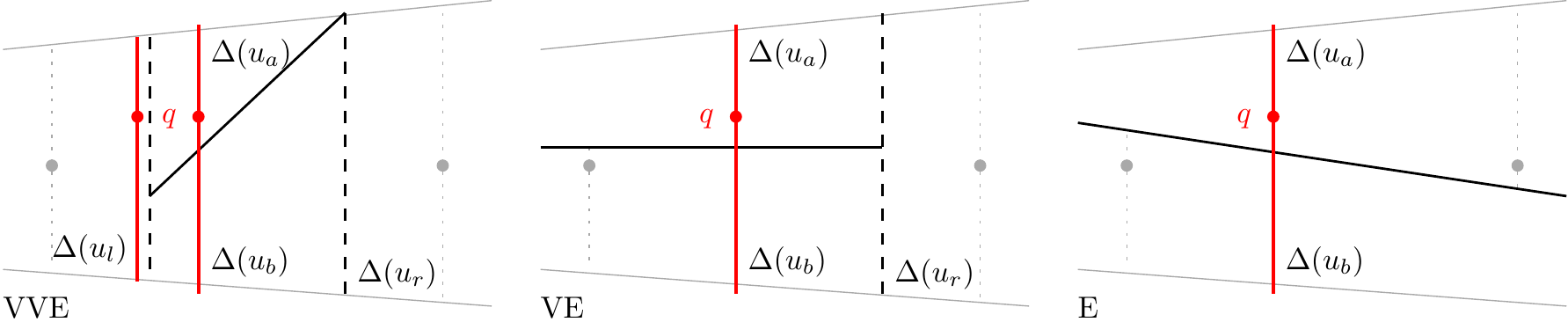}
	\caption{Cases during recursions of $\vPart(u, q,v^-,v^+)$ and $\vMerge(u^-,u^+, q, v)$.} \label{fig:vPart-vMerge-Tree}
\end{figure}

\subsection{Recursive Edge Partitions and Merges}
Based on the recursive primitives for vertical cuts, we introduce the recursive edge partition (c.f. Appendix Algorithm~\ref{alg:partition}) on slab regions that are fully crossed by the edge cut.
See Figure~\ref{fig:Partition-Merge} for an illustration of the (most technical) VVE pattern.

In case of no intersection with the edge cut that destroys $u$, the recursive results $\partition(u_i, c, v^-_i, v^+_i)$ on nodes $u_i$ in $\subtrees(u)$ only need $\vMerge$ calls on one side of $c$
(e.g. left in Figure~\ref{fig:Partition-Merge}: First on $v_a^+,v_r^+$ and then with $v_l^+$).
For new intersections with $c$, we first use $\vPart(u_a,i,v_{al},v_{ar})$ and $\vPart(u_b,i,u_{bl},u_{br})$ for the intersection point $i$, prior to the $\partition$ with $c$.
Finally we use $\vMerge$ to combine those results properly (e.g. right in Figure~\ref{fig:Partition-Merge}: Nodes $v_l^+$ with $v_{al}^+$ and $v_{br}^-$ with $v^-_r$).

Note that for segments from a planar subdivision, the treatment of intersections is not necessary (c.f. Second block of Algorithm~\ref{alg:partition} in Appendix~\ref{sec:code}).
The inverse primitive for edge cut merges (of two adjacent regions) can be derived analogously, by executing the inverse primitives in exactly the reverse order (c.f. Appendix Algorithm~\ref{alg:merge}).

\begin{figure}[] \centering
	\includegraphics[width=\linewidth]{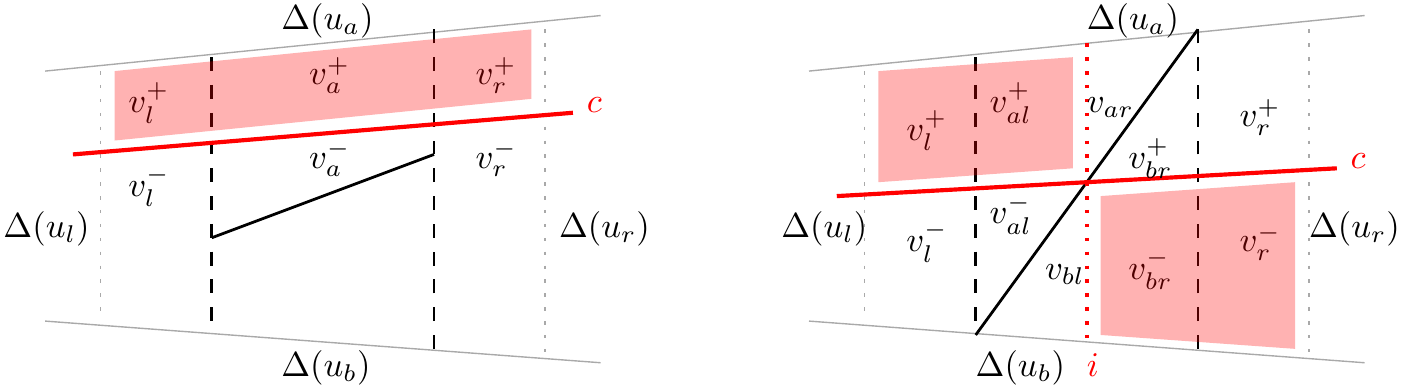}
	\caption{Cases during recursion of $\partition(u, c, v^-, v^+)$ and $\merge(u^-,u^+, c, v)$.} \label{fig:Partition-Merge}
\end{figure}

\section{Counting Search Nodes in Affected Regions} \label{sec:runtime}

To improve readability in this section, we denote for singleton elements the set union $S \cup \{s\}$ with $S + s$ and the set difference $S\setminus\{s\}$ with $S-s$.
With $\Perm(S)$ we denote the set of bijective mappings $\pi : S \to \{1,\ldots,|S|\}$.
For $S' \subseteq S$ we use the predicate ``$\pi(S') \leq |S'|$'' to abbreviate that $\pi(s) \leq |S'|$ holds for each $s \in S'$.

The following definitions concern sets of segments and their (deterministic) induced trapezoidal subdivision $\A$ of the plane.
For $s \in S\setminus B$, we define $F_B(s) \subseteq \R^2$ to denote the region of the faces in $\A(B)$ that are intersected by $s$, that is $F_B(s) = \bigcup_{\{f \in \A(B) : f\cap s \neq \emptyset\}} f$.
Moreover, for $s \in B$ we define the neighborhood region $N_B(s) \subseteq \R^2$ to be the union of all faces of $\A(B)$ that are adjacent to $s$ (e.g. $s$ contributes as edge, vertex or intersection cut).
Note that every point in $\R^2$ is in at most $4$ segment neighborhood regions and 
$N_{B + s}(s) = F_B(s)$.

Given a TST $\T(S,\pi)$ and some priority value $r \in\{1, \ldots, |S|\}$, we consider partitions of its nodes in classes $\T_{<r},\T_{=r}, $ and $\T_{>r}$, having priority less, equal or larger than $r$.
E.g. $\T_{<1}=\emptyset$ and $\T_{>r}$ always contains the leaf nodes.
As introduced in Section \ref{sec:def-tst}, we identify every node $v$ in $\T$ with its associated search region $\Delta(v) \subseteq \R^2$.

From the refinement property (c.f. Section~\ref{sec:def-tsd}), we have that the leafs' partition
$\{{\Delta(v) \subseteq \R^2} : v \in \T(S,\pi),~ p(v) = \infty\}$
of the domain is a (set theoretic) refinement of $\A(S)$, for any $S$ and $\pi \in \Perm(S)$.
This applies in particular for a set of segments $S_{\leq r} = \{s \in S:\pi(s)\leq r\}$.
Since additional leaf-insertions only refine further, we have that the region $\Delta(v)$ of a node $v \in \T_{> r}$ is either contained in or disjoint from neighborhood region $N_{S_{\leq r}}(s)$ for any $s \in S_{\leq r}$.

This provides the following $4$-ply covering bound.
For any $B \subseteq S$ with a $\pi \in \Perm(S)$ such that $\pi(B) \leq |B|$, we have that
\begin{align} \label{eq:refinement}
\sum_{s \in B} \Big|\{ v \in \T_{=l}(S,\pi) ~:~  \Delta(v) \subseteq N_B(s) \}\Big|
~\leq~ 4 \Big|T_{=l}(S,\pi)\Big|~,
\end{align}
for each $l\in\{|B|+1,\ldots,|S|\} + \infty$.

\begin{definition}[Affected Nodes]
Given a segment $s \in S$ of designated priority rank $1\leq r \leq |S|$, we call a node $v$ of $\T_{\geq r}$ over $S-s$ affected if and only if $\Delta(v) \cap s \neq \emptyset$ and it is topmost in $\T$.
That is, $v$ has no parent $u$ in $\T_{\geq r}$ with $p(u) \leq p(v)$.
\end{definition}
Clearly, for every $v \in \T_{\geq r}$ whose region $\Delta(v)$ intersects $s$, we have $\Delta(v) \subseteq F_{S_{<r}}(s)$ as well.
In other words, the affected nodes correspond precisely to the leafs of $\T( S_{<r},\pi)$, the tree over the first $r-1$ segments, that are intersected by $s$.
See Figure~\ref{fig:tree-restrict} for an illustration and Figure~\ref{fig:zone-refine} for an example.
In this example decomposition $\A(s_1,s_2)$ has $7$ faces of which $s_3$ intersects $3$.
The neighborhood region of $s_3$ is $F_{\{s_1,s_2\}}(s_3)=N_{\{s_1,s_2,s_3\}}(s_3)$ and shaded in red.
The TST for 
$\Big(\{s_1,s_2,s_4\},\bigl(\begin{smallmatrix}
s_1 & s_2 	& s_4 \\
1 	& 2 	& 3 
\end{smallmatrix}\bigr)\Big)$ 
has $13$ leaf regions and the TSD has $10$ leafs.
Note that for TST nodes $v$ with priority $p(v)\geq 3$, the region $\Delta(v)$ is either fully contained or outside the red zone. 

\begin{figure}[]
\centering
\begin{subfigure}[t]{0.425\textwidth}
	\centering
	\includegraphics[width=\columnwidth]{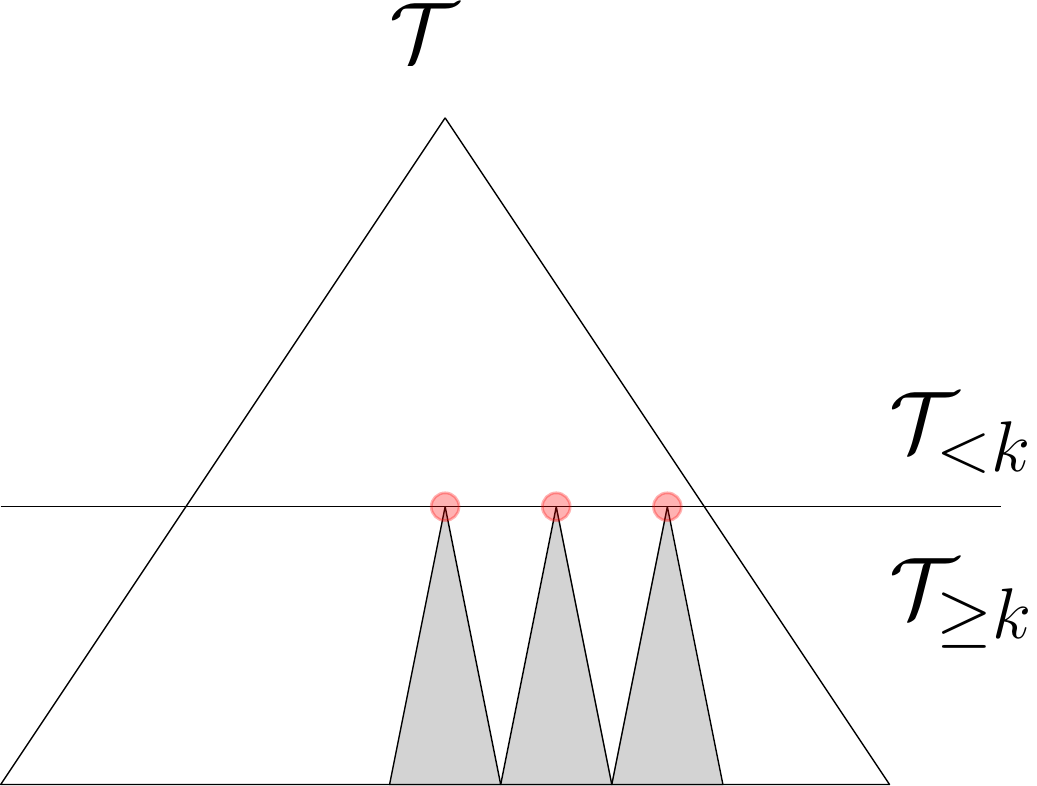}
	\caption{Partition of nodes in $\T$ in $\T_{<k}$ and \quad\\ $\T_{\geq k}$ due  to node priority. Affected \quad\\nodes are in red and their subtrees in gray.}\label{fig:tree-restrict}
\end{subfigure}%
~ 
\begin{subfigure}[t]{0.565\textwidth}
	\centering
	\includegraphics[width=\columnwidth]{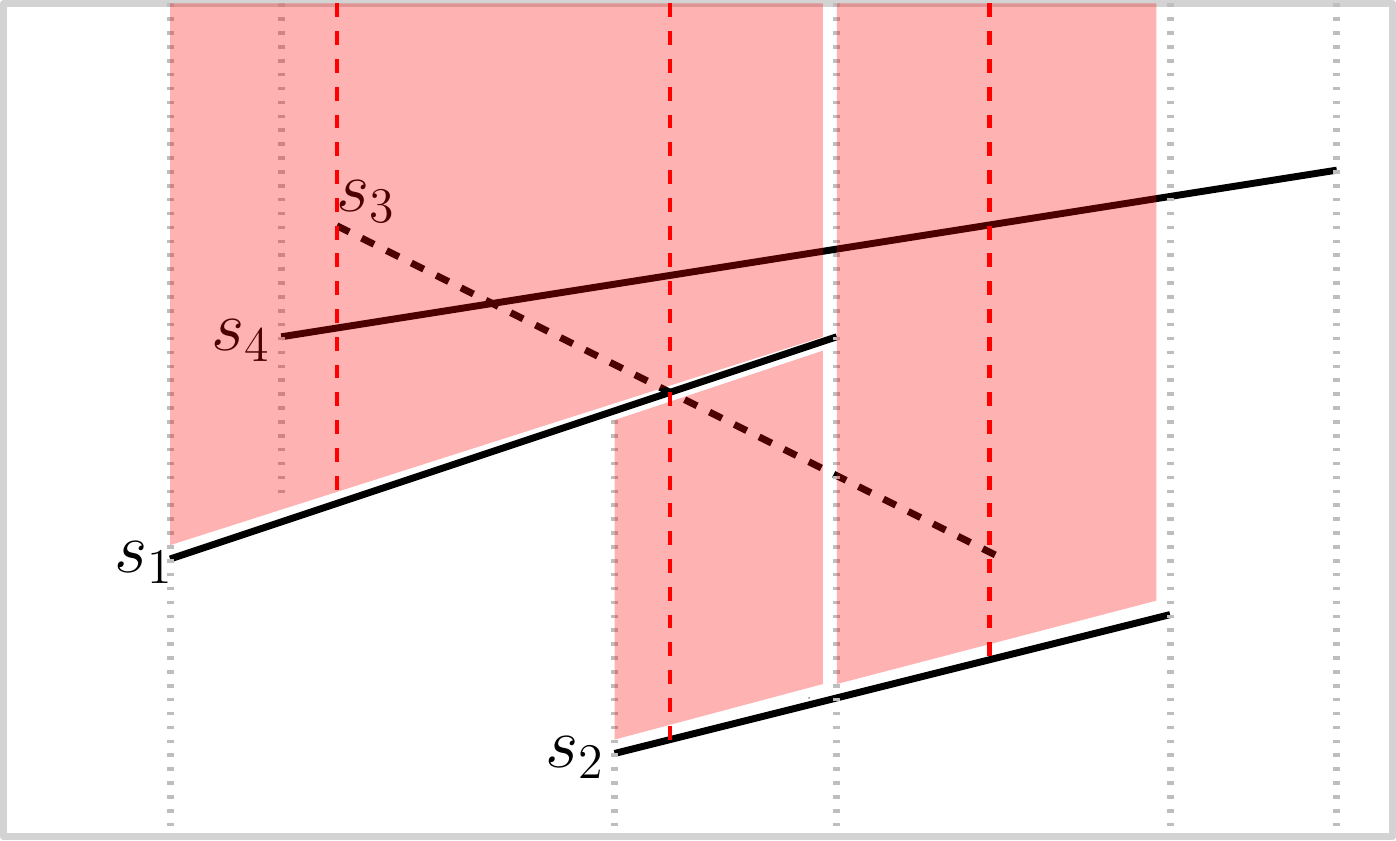}
	\caption{Neighborhood Zone Refinement. The neighborhood region of $s_3$ (shaded in red) is $F_{\{s_1,s_2\}}(s_3)=N_{\{s_1,s_2,s_3\}}(s_3)$.} \label{fig:zone-refine}
\end{subfigure}
\end{figure}

\begin{lemma}[Zone Covering] \label{lem:bucketing}
	Let $S$ be a set of segments,  $1 \leq p \leq |S|$ fixed, and $s \in S$.
	For any $l \geq p$, the expected number of priority $l$ nodes in a random TST over $S-s$, that fall into the neighborhood region of $s$ having priority rank $p$, is upper bounded by
	$$
	\avg_{\pi \in \Perm(S-s)} \Big| \Big\{ v \in \T_{=l}(S-s,\pi)~: \Delta(v) \subseteq F \Big\}\Big|
	~\leq ~
	\frac{4}{p} \avg_{\pi \in \Perm(S-s)} \Big|\T_{=l}(S-s, \pi)\Big| ~,$$
	where $F=F_{\{s' \in (S-s) ~:~ \pi(s')<p\}}(s)$ denotes the neighborhood region of $s$ among the segments of $S-s$ with lower priority than $p$.
\end{lemma}

On a sequence of elements, we call the process to place a new element between the elements of positions $p-1$ and $p$ \emph{emplacing} at $p$ and \emph{dropping} is the reverse operation.
E.g. old elements of index less than $p$ remain and those of at least $p$ are moved one position to the right.
We perform a backward analysis and consider dropping the element at position $p$ in sequences of $\Perm(S)$ instead of emplacing a $s \in S$ at position $p$ in sequences of $\Perm(S-s)$.

\begin{proof}
	We group the sequences in $\Perm(S)$ by those having the same sets $B \in \binom{S}{p}$ in the first $p$ entries to use that $\bigcup_{s \in B} N_B(s) = \R^2$ with (\ref{eq:refinement}).
	
	For $\pi \in \Perm(S)$, we denote with $\pi^{-1}(p) \in S$ the element that is mapped to position $p$ in the sequence and with $\pi -s$ we denote the sequence of $\Perm(S-s)$ that results from dropping element $s$.
	To shorten notation, we define a random variable $\gamma$ that counts the relevant nodes of a tree within a certain region $N \subseteq \R^2$, that is $\gamma(S,\pi,N) = |\{ v \in T_{=l}(S,\pi)~:~\Delta(v) \subseteq N \}|$.
	
	We count the expected number of nodes within $\T(S-\pi^{-1}(p), \pi - \pi^{-1}(p))$, that fall into the region $N_{\{s \in S~:~\pi\left(s\right)\leq p\}} \left(\pi^{-1}\left(p\right) \right)$, which is the neighborhood of the element at position $p$ among those with priority at most $p$.
	\begin{align}
	&
	\avg_{\pi \in \Perm(S)} 
	\gamma\left(
	S-\pi^{-1}\left(p\right), 
	\pi - \pi^{-1}\left(p\right),
	N_{\{s \in S~:~\pi\left(s\right)\leq p\}} \left(\pi^{-1}\left(p\right) \right)
	\right)
	\\
	=&
	\frac{1}{|S|!}\sum_{B \in \binom{S}{p}} \sum_{s \in B}
	\sum_{\overset{\pi \in \Perm(S)~:}{\pi(B)\leq p,~\pi(s)=p}}
	\gamma(S-s,\pi-s, N_B(s) )
	\label{eq:BWD-zone-regroup}
	\\
	=&
	\frac{1}{|S|!}\sum_{B \in \binom{S}{p}} ~~ \frac{1}{p}\sum_{s \in B}~~ \sum_{s \in B} 
	\sum_{\overset{\pi \in \Perm(S)~:}{\pi(B)\leq p,~\pi(s)=p}}
	\gamma(S-s,\pi-s, N_B(s) )
	\label{eq:BWD-zone-doublecount}
	\\
	\leq&
	\frac{1}{|S|!}
	\sum_{B \in \binom{S}{p}}
	\sum_{s \in B} \frac{1}{p}
	\sum_{\overset{\pi \in \Perm(S)~:}{\pi(B)\leq p,~\pi(s)=p}}
	4\cdot\gamma(S-s,\pi-s,\R^2 )
	\label{eq:BWD-4-ply}
	\\
	=&
	\frac{4}{|S|!~p}
	\sum_{s \in S}
	\sum_{B \in \binom{S-s}{p-1}}
	\sum_{\overset{\pi \in \Perm(S-s)~:}{\pi(B) < p}}
	\gamma(S-s,\pi,\R^2 )
	\label{eq:BWD-comb-exchange}
	\\
	=&
	\frac{1}{|S|}
	\sum_{s \in S}
	\frac{4}{p}
	\frac{1}{|S-s|!}
	\sum_{\pi \in \Perm(S-s)}  \gamma(S-s,\pi,\R^2 )
	\\
	=&
	\frac{1}{|S|}
	\sum_{s \in S} ~~
	\frac{4}{p}
	\avg_{\pi \in \Perm(S-s)}  \gamma(S-s,\pi,\R^2 ) \label{eq:ZoneLastLine}
	\end{align}
Equation (\ref{eq:BWD-zone-regroup}) is due to regrouping the summation terms,  (\ref{eq:BWD-zone-doublecount}) due to double counting, and (\ref{eq:BWD-4-ply}) due to the $4$-ply covering.
Equation (\ref{eq:BWD-comb-exchange}) uses the combinatorial identity $\binom{n}{k}\binom{k}{1}=\binom{n}{1}\binom{n-1}{k-1}$ to first choose $s\in S$ for position $p$.
Since the terms in (\ref{eq:ZoneLastLine}) do not depend on the spatial location of $s$, the bound holds for any $s\in S$.
\end{proof}

We now bound the expected total size of subtrees below affected nodes (c.f. Figure~\ref{fig:tree-restrict}).
The proof of the following statement hinges on a combinatorial identity that counts the different sequences of a process that puts $n$ elements in an array with $n$ position slots.
$$
\sum_{p=1}^n \binom{n}{1} \binom{n-1}{n-p}(p-1)!(n-p)!  \quad = \quad n \cdot n!
$$
Though there are $n \cdot n!$ sequences to the processes, each of the resulting $n!$ assignments is created exactly $n$ times. Hence also the cost associated with one such assignment.

\begin{lemma}[Subtree Sizes] \label{lem:affected-subtrees}
	Let $S$ be a set of segments and $s \in S$ be a fixed element with a designated, uniformly random priority in $\{1, \ldots, |S|\}$.
	The expected total size of affected subtrees in a random TST over $(S-s)$ is $\O \left(\log^2 |S| + \frac{k_S}{|S|} \log|S| \right)$, where $k_S$ denotes the number of intersecting pairs in $S$.
\end{lemma}
\begin{proof}
We express the expected cost in terms of the left hand side of the above combinatorial identity, which allows us to use Lemma~\ref{lem:bucketing} for all nodes that have at least the designated priority of $s$.
\begin{align}
&~
\frac{1}{|S|^2}\sum_{1 \leq p \leq |S|} 
\sum_{s \in S } 
\frac{4}{p} \avg_{\pi \in \Perm(S-s)} \Big|\T_{\geq p}(S-s, \pi)\Big|
\\
\leq&~
\frac{1}{|S|^2}\sum_{1 \leq p \leq |S|} 
\sum_{s \in S } 
\frac{4}{p} \avg_{\pi \in \Perm(S-s)} \Big|\T(S-s, \pi)\Big|
\\
\leq&~
\frac{1}{|S|^2}\sum_{1 \leq p \leq |S|} 
\frac{4}{p}
\sum_{s \in S } 
\O \left( |S-s|\log |S-s| + k_{S-s} \right) \label{eq:subTreeSizeTotalSize}
\\
\leq&~
\O \left(
\frac{1}{|S|^2}\sum_{1 \leq p \leq |S|} 
\frac{|S||S-s|\log |S-s| + (|S|-1)k_S}{p}
\right)
\leq \O \left(\log^2|S|+\frac{k_S}{|S|}\log |S| \right)
\label{eq:subTreeSizeIntersect}
\end{align}
Bound (\ref{eq:subTreeSizeTotalSize}) is due to the expected size of the whole tree (c.f. Theorem~\ref{thm:exp-tree-size}) and (\ref{eq:subTreeSizeIntersect}) is due to $\sum_{s\in S} k_{S-s}=\sum_{s\in S} (k_S-k_{\{s\}})=(|S|-1)k_S$.
\end{proof}

We now bound the cost for the iterated ray-shooting search to find the affected subtree roots for a query segment that has a designated random priority rank (c.f. Section~\ref{sec:segment-search}).

\begin{lemma}[Segment Search] \label{lem:segment-search}
	Let $S$ be a set of non-crossing segments and $s \in S$.
The expected time to find the affected nodes within a random TST over $(S-s)$	
 is $\O \left(\log^2 |S|\right)$.
\end{lemma}

\begin{proof}
	The search visits only nodes $v$ with regions $\Delta(v)\cap s \neq \emptyset$.
	Since Lemma~\ref{lem:affected-subtrees} bounds the size of the result set, we only need to bound the number of nodes with priority smaller $p(s)$.
	We do so by invoking Lemma~\ref{lem:bucketing} $\lceil \log_ 2(p(s))\rceil$ times.
	\begin{align}
&\frac{1}{|S|}\sum_{p=1}^{|S|}
\avg_{\pi \in \Perm(S-s)}
\Big| \Big\{ v \in \T(S-s,\pi)~:~ p(v)\leq p ,~ \Delta(v) \cap s \neq \emptyset\Big\} \Big|
\\
\leq &
\frac{1}{|S|}\sum_{p=1}^{|S|} \sum_{j=1}^{\lceil \log_2 p\rceil} 
\avg_{\pi \in \Perm(S-s)}
\Big| \Big\{ v \in \T(S-s,\pi)~:~ 2^{j-1}<p(v)\leq 2^j ,~ \Delta(v) \cap s \neq \emptyset\Big\} \Big|
\label{eq:segSearch-batches}
\\
\leq &
\frac{1}{|S|}\sum_{p=1}^{|S|} \sum_{j=1}^{\lceil \log_2 p \rceil} 
\frac{4}{2^{j-1}}
\avg_{\pi \in \Perm(S-s)}
\Big| \Big\{ v \in \T(S-s,\pi)~:~ 2^{j-1}<p(v)\leq 2^j \Big\} \Big|
\label{eq:segSearch-lemma}
\\
\leq &
\frac{1}{|S|}\sum_{p=1}^{|S|} \sum_{j=1}^{\lceil \log_2 p \rceil} 
\frac{4}{2^{j-1}}
2^j \log (2^j)
\leq
\frac{1}{|S|}\sum_{p=1}^{|S|} \sum_{j=1}^{\lceil \log_2 p \rceil} 
8 j \log(2)
\\
\leq &
\frac{1}{|S|}\sum_{p=1}^{|S|} \O(\log^2 p)
\leq 
\frac{ \O(|S|\log^2 |S|)}{|S|}
\leq \O(\log^2 |S|)
  \end{align}
	
With (\ref{eq:segSearch-batches}), we count search nodes in batches of priority intervals of the form $(2^{j-1}, 2^j]$.
Bound (\ref{eq:segSearch-lemma}) is due to invoking Lemma~\ref{lem:bucketing} for $s$ according to the priority interval.
	For such a batch $j$, we bound $\avg_{\pi \in \Perm(S-s)}|\T_{\leq 2^j}(S-s, \pi)|$ with $\O(2^j \log (2^j))$, that is the expected size of a TST over a set of $2^j$ segments.
\end{proof}

Now we have all necessary arguments for our bound on dynamic update costs in TSTs.
\begin{theorem} \label{thm:TSTinsert}
	Let $S$ be a set of non-crossing segments and $s \in S$.
	The expected cost of inserting $s$ in a random TST over $(S-s)$ 
	is $\O(\log^2 |S|)$.
\end{theorem}
\begin{proof}
Let $1\leq m \leq |\T|$ denote the total number of nodes within affected subtrees of $\T$, prior to the insertion call.
We argue that the total number of nodes visited by the insertion procedure is in $\O(m)$, which establishes the result based on Lemma~\ref{lem:affected-subtrees} and \ref{lem:segment-search} for non-crossing segments ($k_S=0$).

For each root of an affected subtree, we have at most $2$ calls of $\vPart$ to slab the subtree.
After the slabbing stage of the insertion, we have at most $3m$ nodes in total.

The edge $\partition$ calls on one slab are independent from those of another slab in $\T$.
Neglecting node removals due to $\vMerge$ briefly, the total number of temporarily created nodes of $\partition$ increases at most by a factor of $2$. 
Since we cannot remove more than what was created, the total number of nodes that $\vMerge$ may visit is bounded by $6m$.
\end{proof}

Given the duality of the update procedures, deletion of a $s \in S$ having priority rank $p(s)$ visits exactly the same nodes as insertion of $s$ among $(S-s)$ with priority rank $p(s)$.
Hence the expected node visits are equal as well.
Since our bound on expected update operations counts TST search nodes, 
we employ the result of Hemmer et al.~\cite{Hemmer16}, on the bijection between TST and TSD search paths, to obtain equal asymptotic bounds for the expected update costs in TSDs.

\section{Offline vs Online -- Maintaining Small Codes For Dynamic Orders}\label{sec:online-vs-offlie}
To support online emplacement and dropping of additional elements in a sequence $\pi \in \Perm(S)$, we use a more flexible representation than standard bijections $\pi : {S \to \{1,\ldots,|S|\}}$.
We represent sequences as injective mappings $\tau : S \to (0,1) \subseteq \R$.

If the values of $\tau$ are stored with the elements of $S$, we can evaluate the required $\pi(s) < \pi(s')$ predicates with comparison on $\tau$ in constant time.
However, dynamically extending $\tau$ for a new element $s \notin S$ such that $\tau(s)$ falls with equal probability in the intervals between $\{\tau(s'):s'\in S\}$ is challenging.
E.g. simple random sampling of $\tau(s)$ uniformly out of $(0,1)$ does not provide this. 
Since single registers of the pointer machine model are confined to numbers with only $\O(\log |S|)$ bits, simple interval halving strategies may well produce inefficiently long codes for $\tau$ as well.

We solve this problem using one additional randomized data structure to maintain orders among $n$ elements of a totally ordered set of `keys'.
Treaps are a conceptual fusion of Binary Search Trees (BSTs), over the nodes' key values, and (Min-)Heaps, over the nodes' priority values.
Insertions and deletions are performed similarly to BSTs and additionally standard tree rotations are used to maintain the heap property on the randomly chosen priority values.
Seidel and Aragon~\cite{SeidelA96} show not only that Treaps achieve $\O(\log n)$ depth with high probability but also that, even if the cost of a rotation is proportional to the whole subtree size, insertions and deletions still perform expected $\O(\log n)$ operations.

For our simple in-order numbering scheme in a fixed BST, we consider the injective mapping from root-to-node paths, which we read as smaller/larger strings over the alphabet $\{s,l\}$, to binary fractional number in $(0,1)\subseteq \Q$, which is a `$1$'-terminated string over the alphabet $\{0,1\}$.
More precisely, for the empty path string (the root node $r$) the associated value is $\gamma(r)=0.1_{(2)}$, $s$-edges are coded as $0$ and $l$-edges are coded as $1$.
E.g. the left child node $v_s$ of the root has code $\gamma(v_s)=0.01_{(2)}$, the right child node $v_l$ of the root has code $\gamma(v_l)=0.11_{(2)}$, and we have $\gamma(v_s) < \gamma(r) < \gamma(v_l)$.
On paths of a BST, the search tree property extends to the order of these $\gamma$ values of the nodes. This allows us to check if $u$ is before $v$ in the in-order sequence by comparing the values of $\gamma(u)$ and $\gamma(v)$.
Moreover, this recursive code definition can be assigned in a top-down fashion for all nodes in a subtree under $v$, once $\gamma(v)$ is assigned.

In order to maintain the proper values $\tau$ for a set $S$, we augment the Treap nodes to also store the total number of nodes in their subtree.
To emplace a new element, we first choose an integer $1 \leq p \leq |S|+1$ uniformly at random.
Next we use the standard Treap insertion to add a new node resembling the $p$-smallest key-value as a new leaf.
After the bottom-up re-balancing with rotations is finished, we simply re-visit the whole rotated subtree in a top-down fashion to overwrite the key value of each node $v$ with the new value $\gamma(v)$.
The deletion of an element is handled analogously.

To evaluate an order predicate between two elements, we simply compare the current key values due to the Treap, which are stored with the elements.
We summarize this with following statement.
\begin{obs}
	Representation of sequences $\pi \in \Perm(S)$ requires no more than $\O(|S|)$ space and dynamic updates take expected $\O(\log |S|)$ operations.
	After successful updates, evaluation of a $\pi(s)<\pi(s')$ predicate on $s,s' \in S$ takes $\O(1)$ operations.
\end{obs}

\vspace{-.2cm}
\section{Implementation and Experiments}\label{sec:exp}
\vspace{-.2cm}
We implemented our fully-dynamic approach for TSTs based on the exact predicates and exact construction for line segments of the  Computational Geometry Algorithms Library~\cite{CGAL}.
The implementation\footnote{\url{https://github.com/milutinB/dynamic_trapezoidal_map_impl}} provides both standard RIC leaf-level insertions and our dynamic version.
In particular, it allows to test a RIC and dynamic TST for equivalence.

Our experiments focus on evaluating the practicality of the approach and the tightness of the analysis.
We use the standard RIC of TSTs as verification and baseline comparison.
To measure the performance of our segment search and dynamic updates, we count the total number of node visits during recursions of the update and the search.
To capture the asymptotic behavior in the experiments, we perform many insertion calls, with either ascending segment priority values (static TST) or a random shuffle of them (dynamic TST).

We created random sets of segments with varying numbers of intersections based on the 64-bit Mersenne Twister random number generator of the C++ Standard Library.
The experiment without intersections is comprised of horizontal segments, for which we first chose a $y$-coordinate uniformly (from the domain range) and then two $x$-coordinates uniformly (c.f. Figure~\ref{fig:exp-non}).
The experiment with few intersections are comprised of short segments, for which we chose uniformly a point, direction, and a length uniformly at random (on average $3\%$ of the domain boundary length; c.f. Figure~\ref{fig:exp-few}).
The experiment on many intersections are generated by choosing the coordinates of both endpoints uniformly at random (c.f. Figure~\ref{fig:exp-many}).

Figure~\ref{fig:dynamic-tree-vs-dag} shows the results on the segment search and recursive node visits of updates.
For ease of comparison, we also plot a function fit\footnote{We use the standard least-squares Marquardt-Levenberg algorithm of GNUPlot 5.2.} (thick blue lines) of the well known size and depth bounds (top rows) and our new segment search and update bounds (bottom rows) as overlay on the experimental data.
Note that only on the non-intersecting data set a comparison of function fit and bounds is meaningful.
We do however also provide it for the other experiments as a `trend line' that indicates the additional costs due to segment intersections.
To allow a better visual perception of the average cost per insertion, Figures~\ref{fig:exp-non} and \ref{fig:exp-few} show every $20$th data point and Figure~\ref{fig:exp-many} shows every $2$nd data point as impulse.

The experiment on non-intersecting segments matches with our analysis of segment search and update operations on TSTs.
As indicated by Lemma~\ref{lem:affected-subtrees}, TST updates on intersecting segments are more expensive.

\begin{figure}[p]
\begin{subfigure}[t]{\linewidth}
	\centering
\includegraphics[width=.4\linewidth]{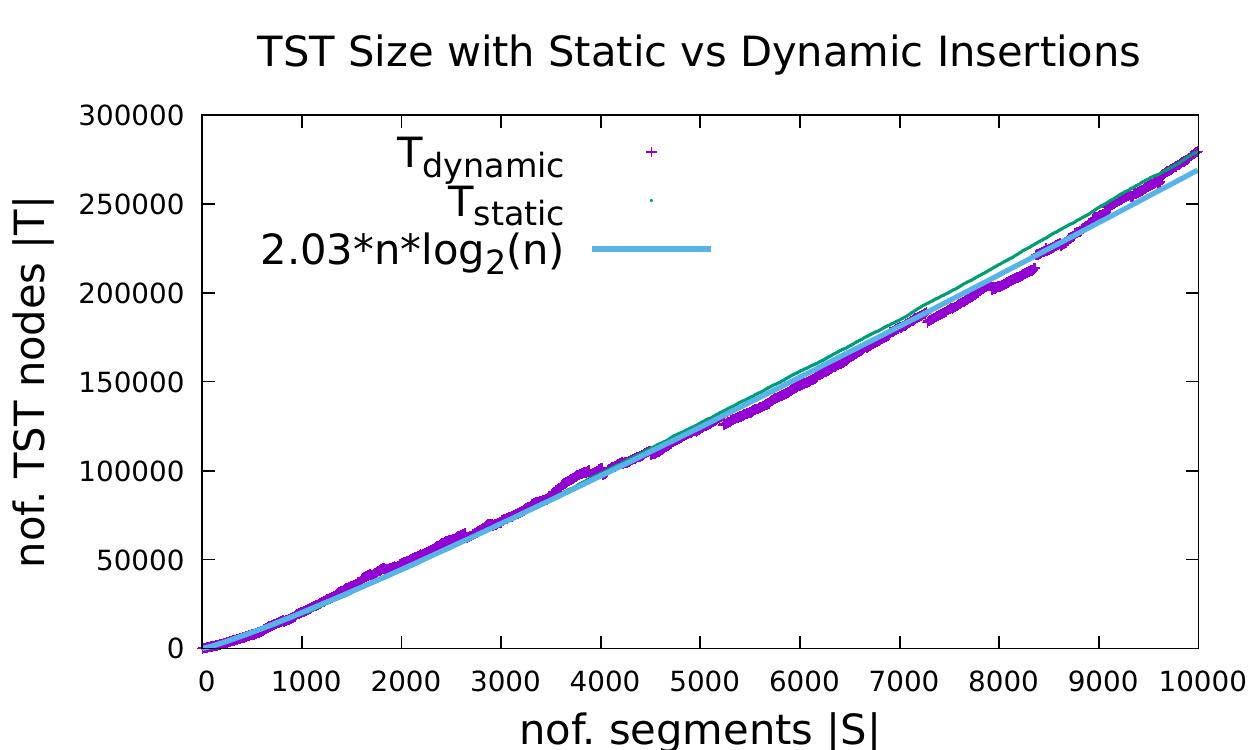}
\includegraphics[width=.4\linewidth]{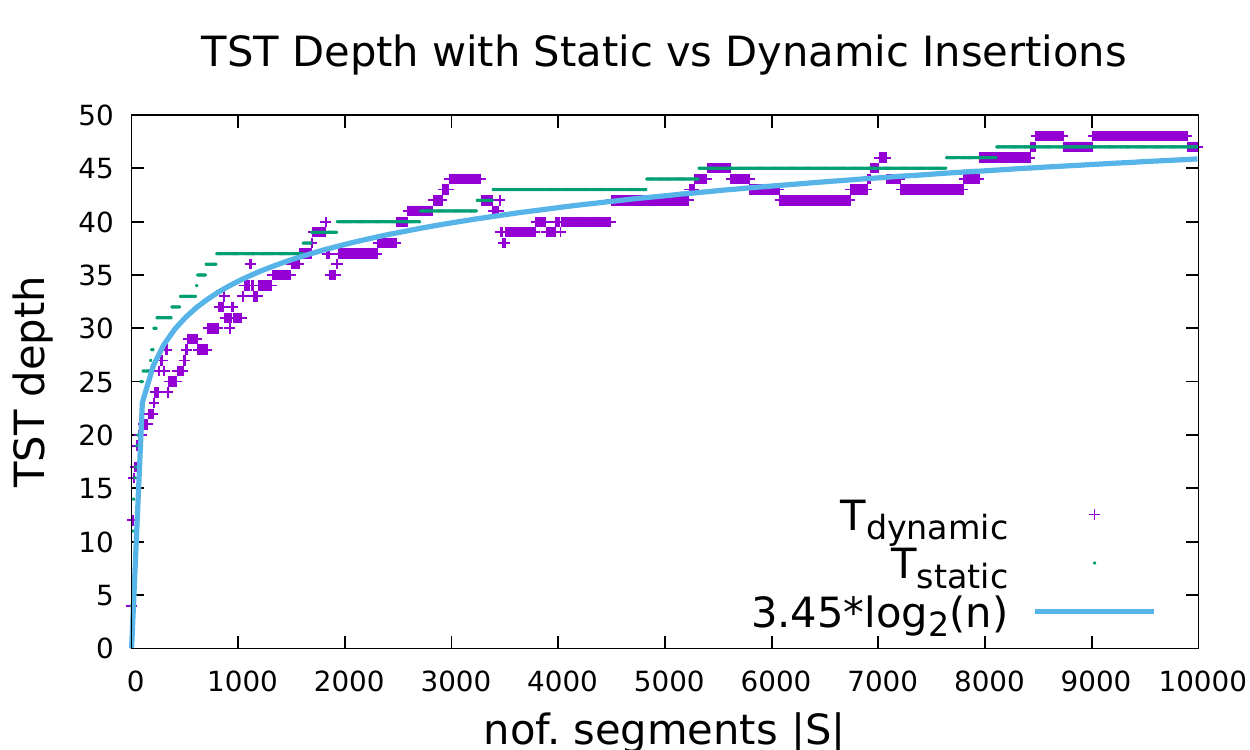}\\
\includegraphics[width=.4\linewidth]{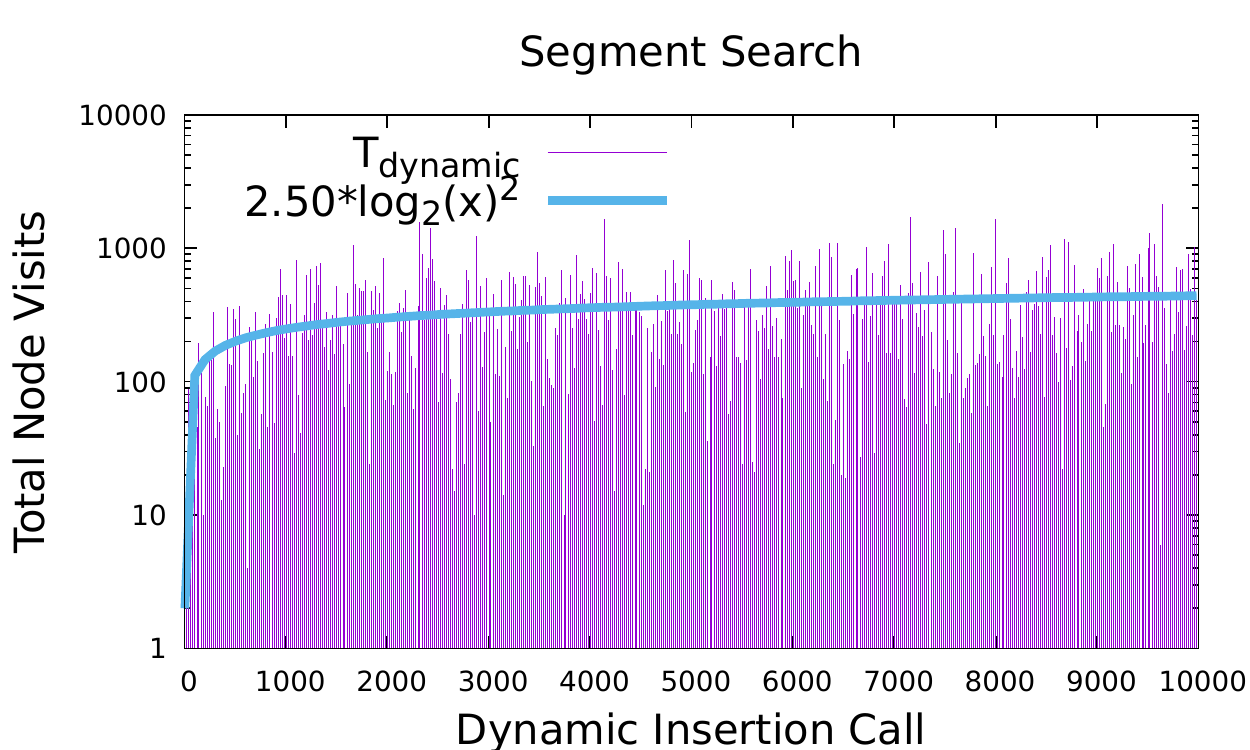}
\includegraphics[width=.4\linewidth]{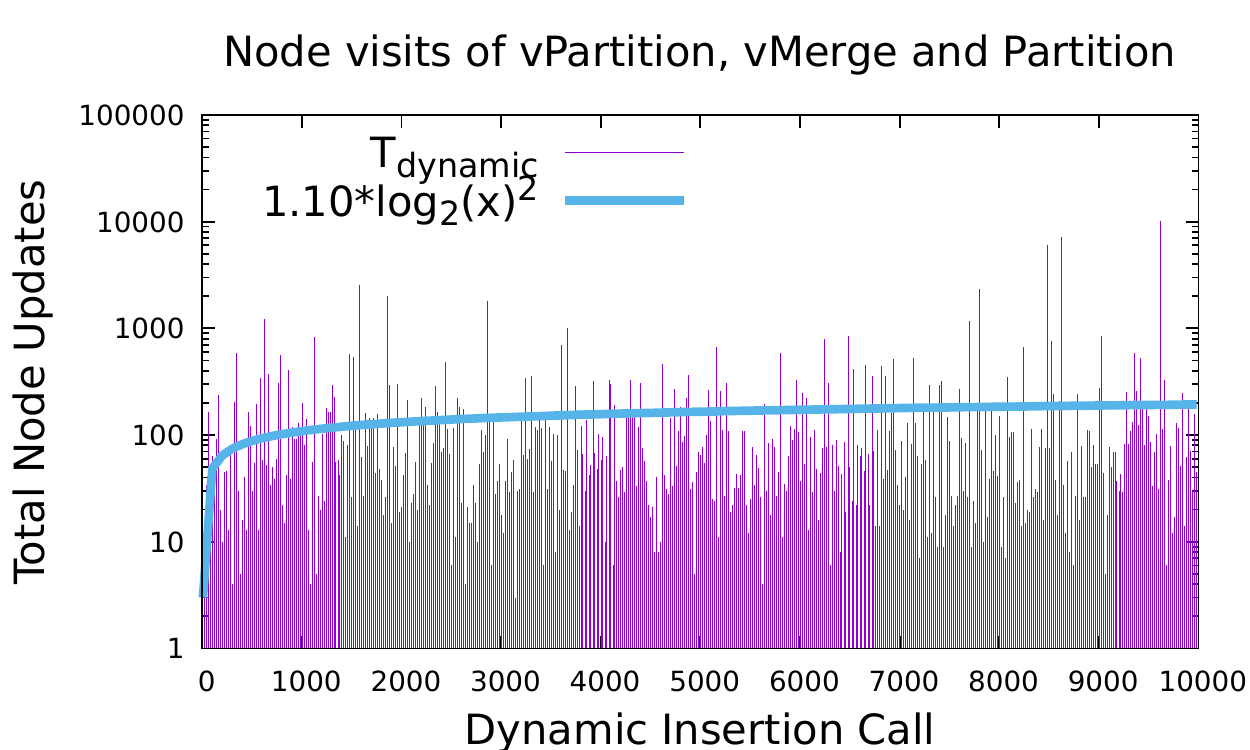}
\caption{Random non-intersecting segments ($|S|=10^4,~ k_S=0,~ k_S/|S|=0$).}
\label{fig:exp-non}
\end{subfigure}%

\begin{subfigure}[t]{\linewidth}
	\centering
	\includegraphics[width=.4\linewidth]{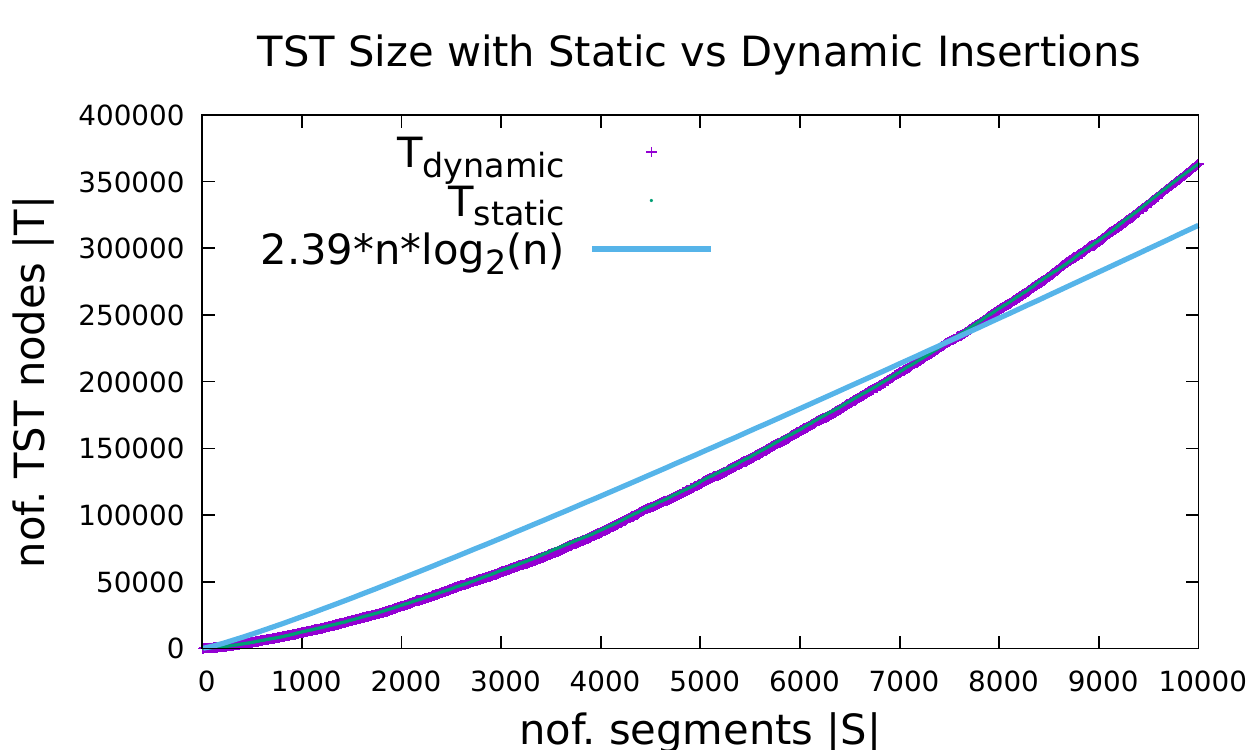}
	\includegraphics[width=.4\linewidth]{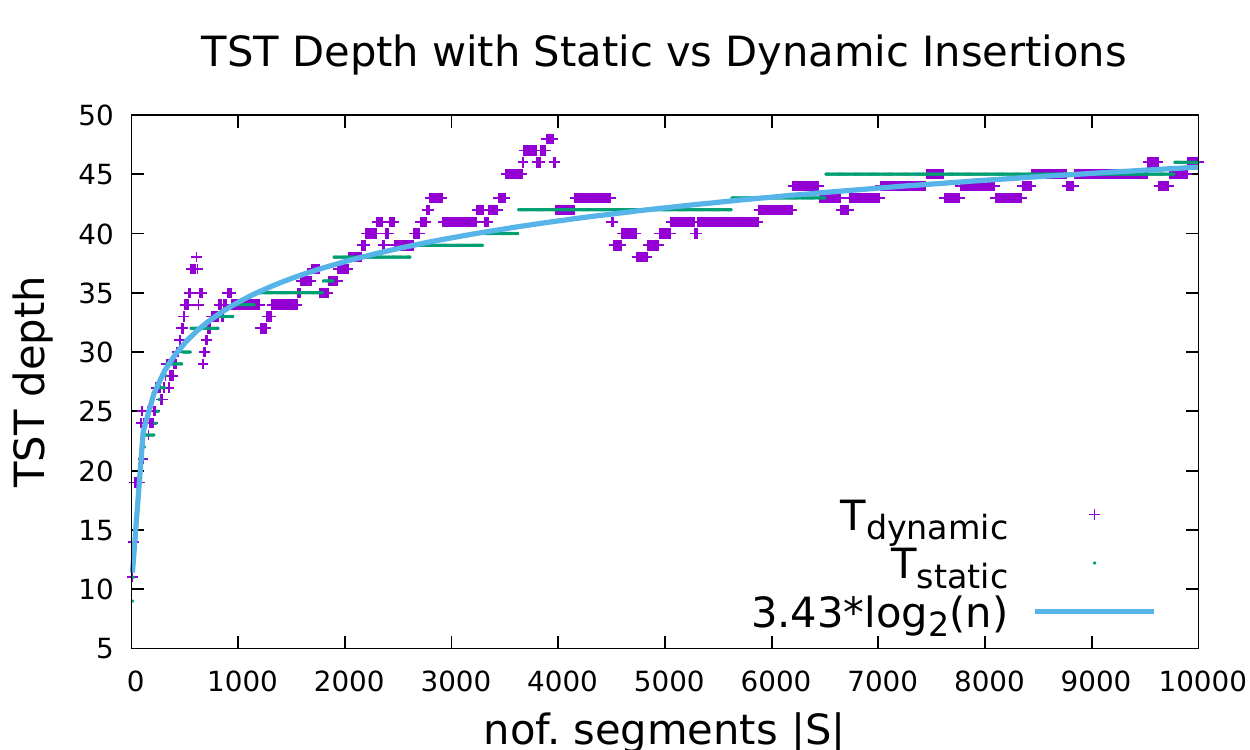}\\
	\includegraphics[width=.4\linewidth]{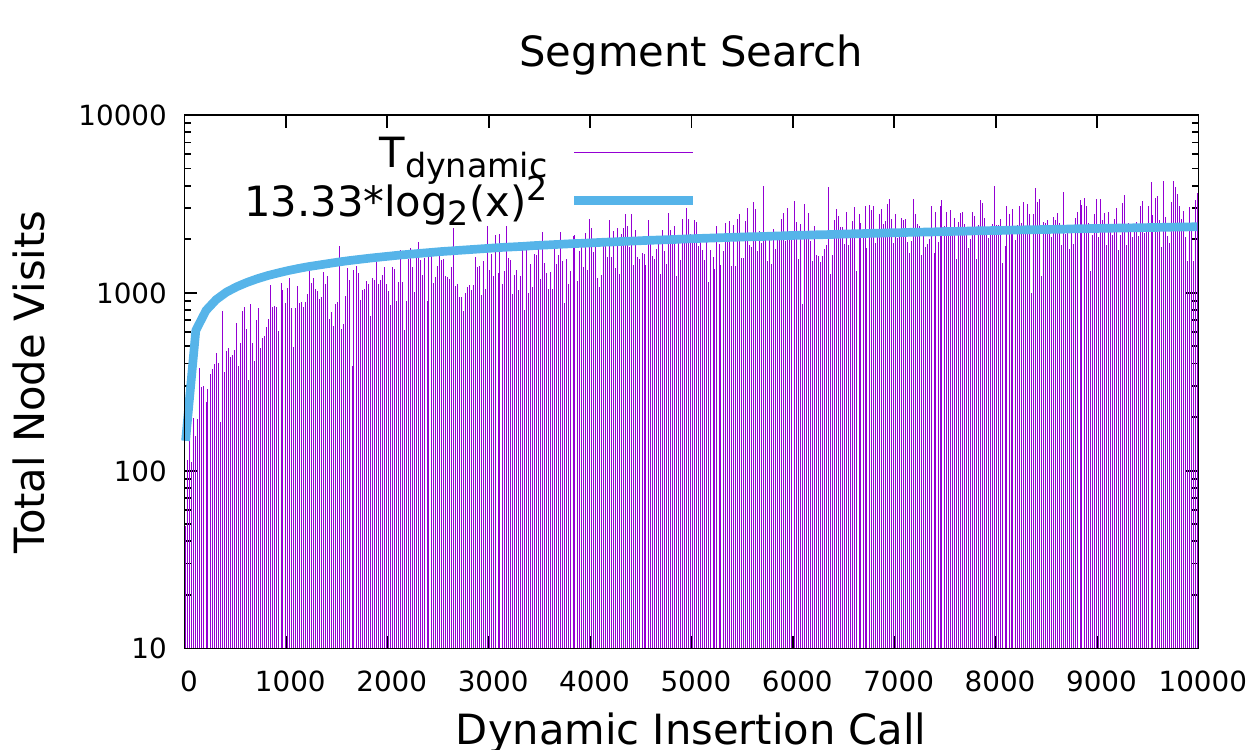}
	\includegraphics[width=.4\linewidth]{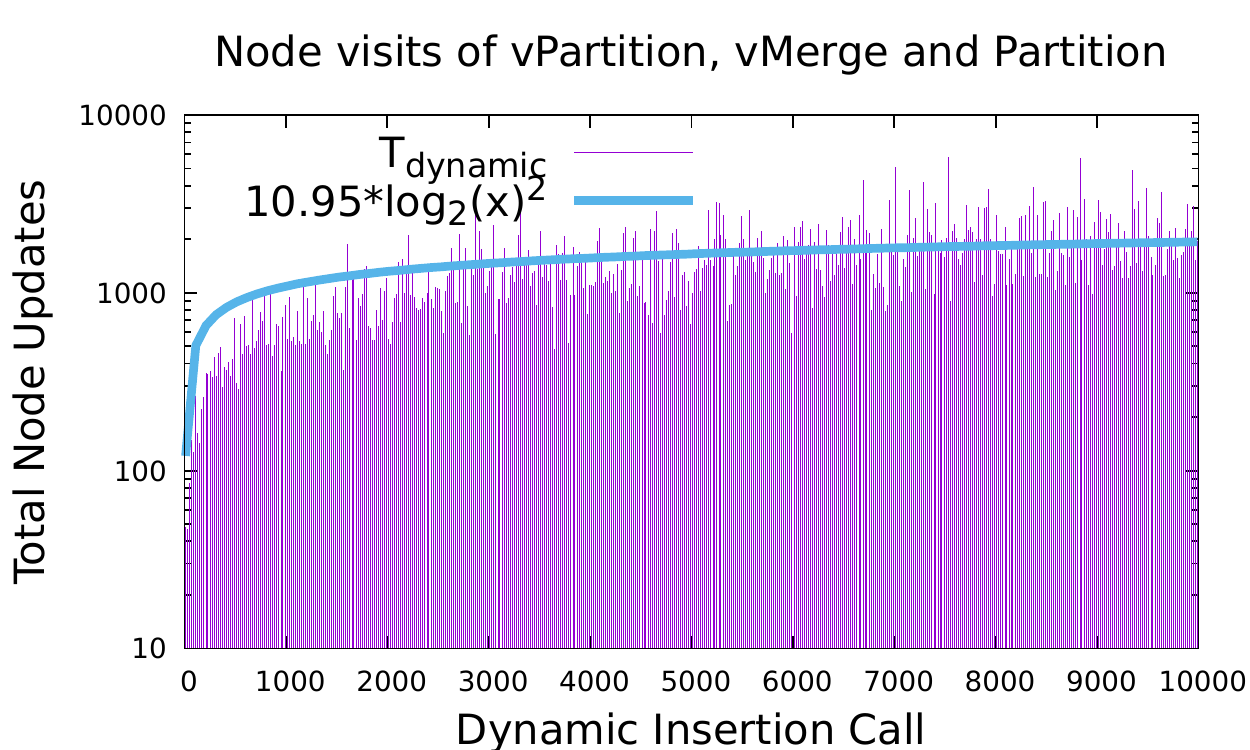}
	\caption{Random segments with few intersections ($|S|=10^4,~ k_S=35195,~ k_S/|S|\approx 3.5$).}
	\label{fig:exp-few}
\end{subfigure}%

\begin{subfigure}[t]{\linewidth}
	\centering
	\includegraphics[width=.4\linewidth]{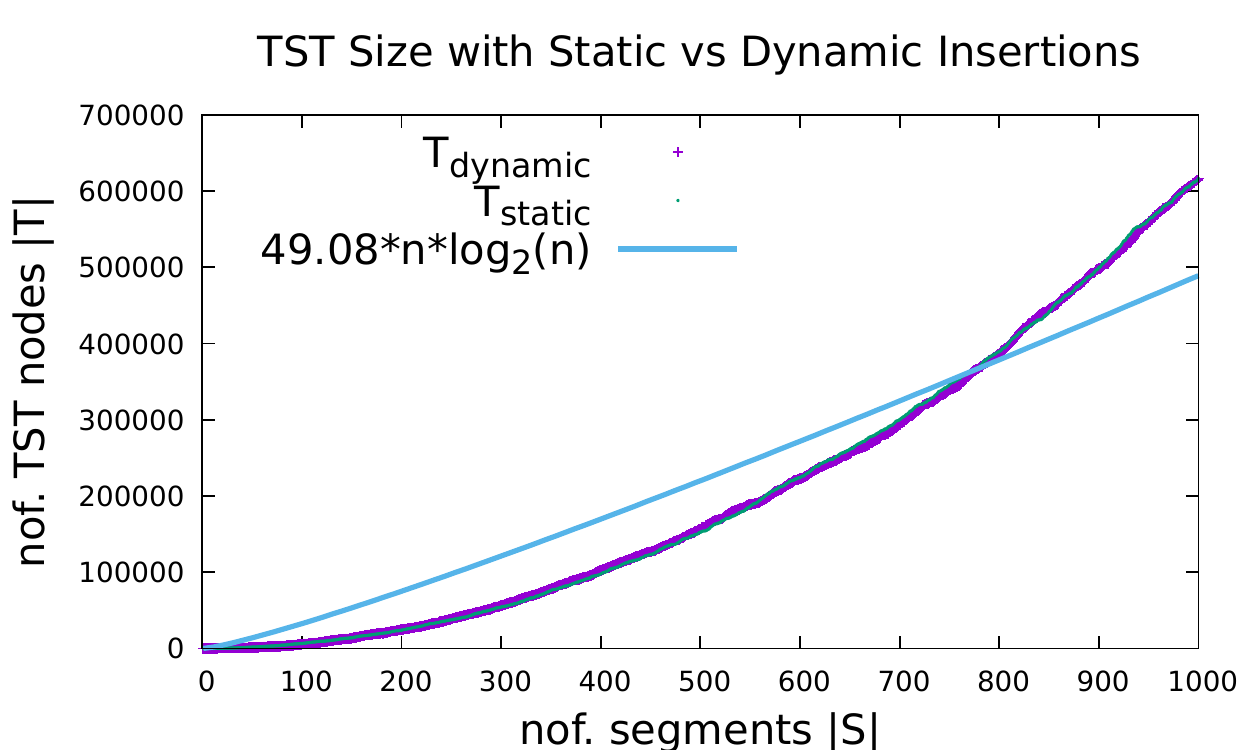}
	\includegraphics[width=.4\linewidth]{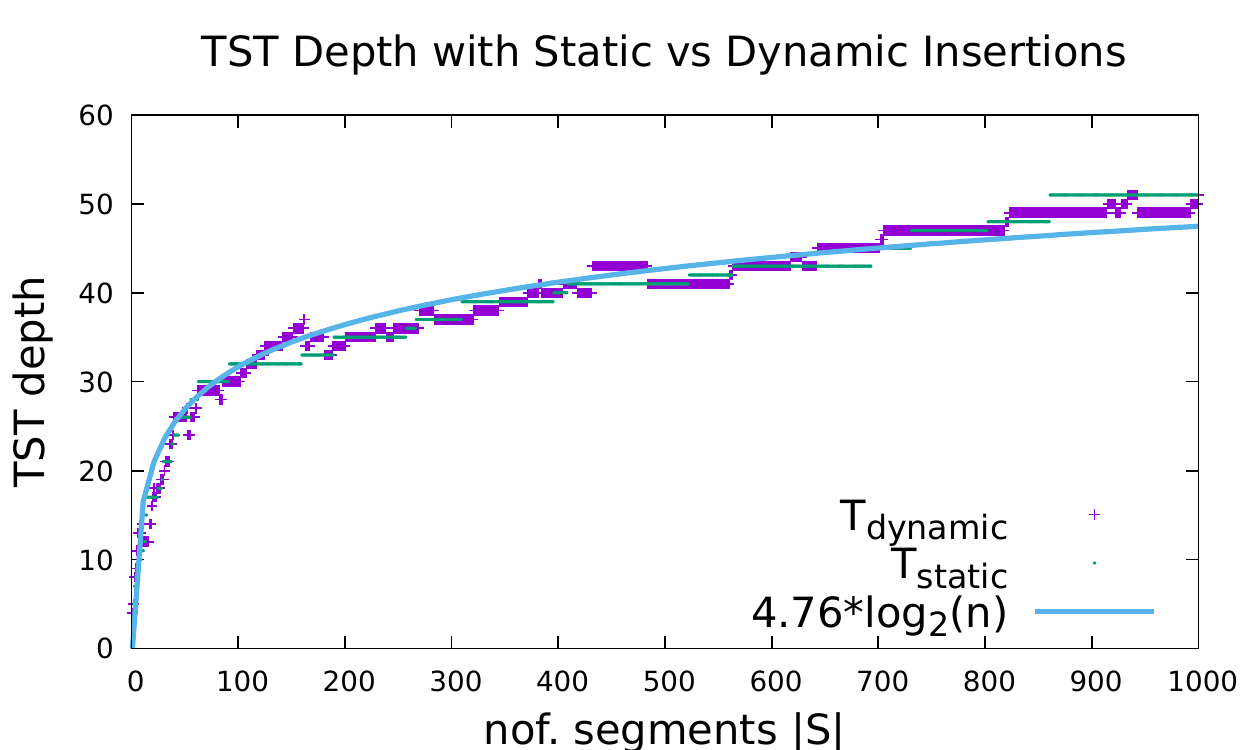}\\
	\includegraphics[width=.4\linewidth]{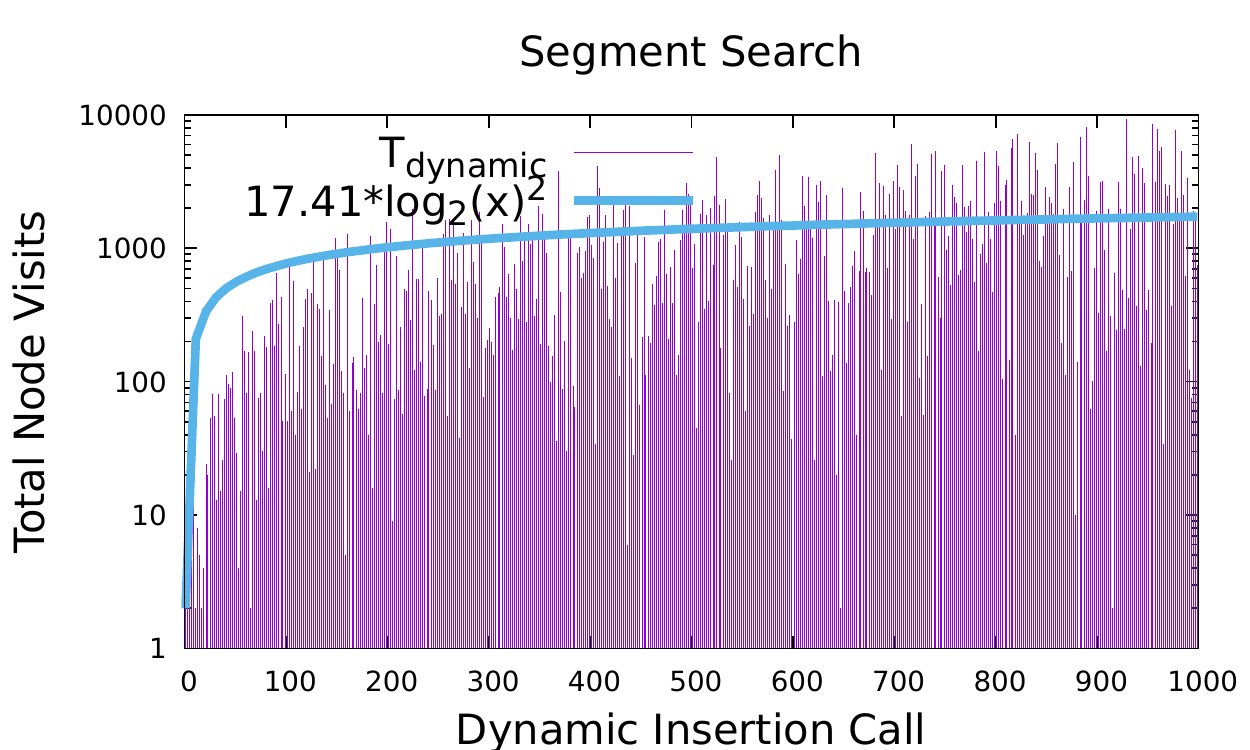}
	\includegraphics[width=.4\linewidth]{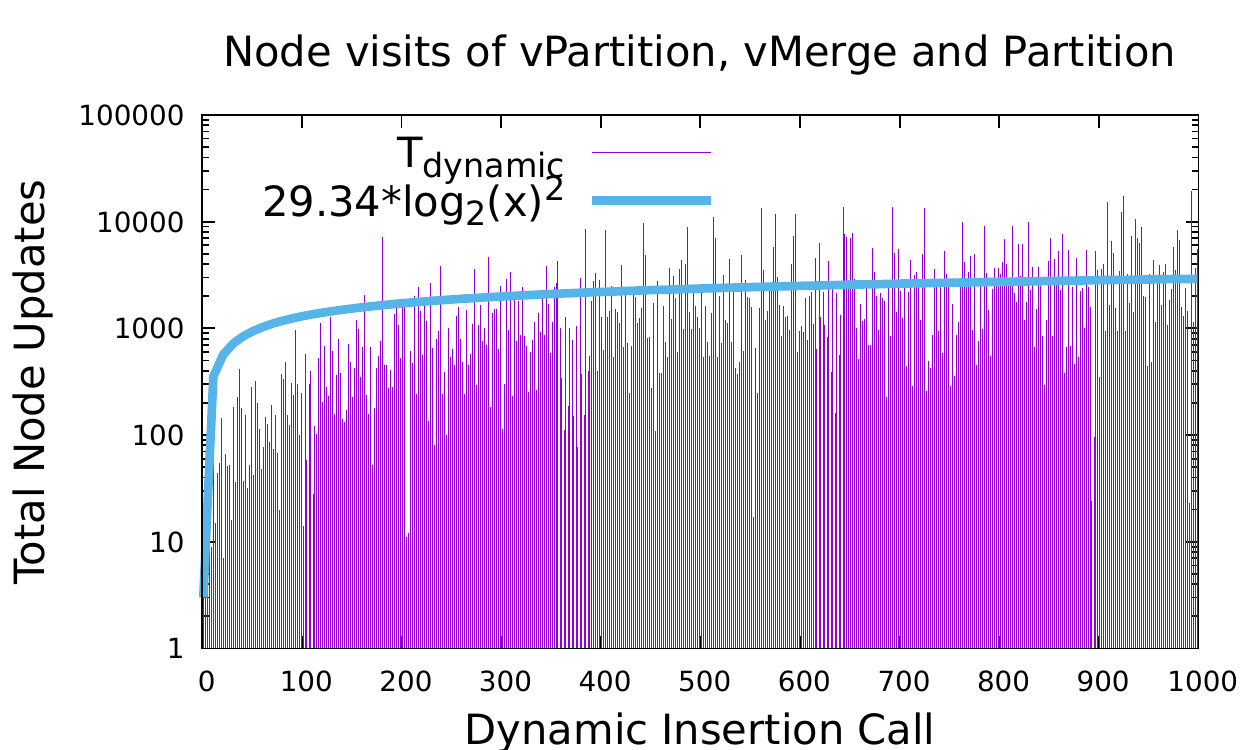}
	\caption{Random segments with many intersections ($|S|=10^3,~ k_S=120730,~ k_S/|S|\approx 120.7$).}
	\label{fig:exp-many}
\end{subfigure}%

\caption{Experimental results on TST size (top left), TST depth (top right), and node visits of the segment search (bottom left) and updates (bottom right) for RIC TSTs (green) and Dynamic TSTs (purple).
The blue curve indicates the respective function fit on the data of dynamic TSTs.
}\label{fig:dynamic-tree-vs-dag}
\end{figure}

\newpage
\bibliographystyle{plainurl}
\bibliography{references.bib}

\begin{thebibliography}{10}

\bibitem{AgarwalBBGH00}
Pankaj~K. Agarwal, Julien Basch, Mark de~Berg, Leonidas~J. Guibas, and John
  Hershberger.
\newblock Lower bounds for kinetic planar subdivisions.
\newblock {\em Discrete {\&} Computational Geometry}, 24(4):721--733, 2000.
\newblock \href {http://dx.doi.org/10.1007/s004540010060}
  {\path{doi:10.1007/s004540010060}}.

\bibitem{AgarwalEG98}
Pankaj~K. Agarwal, Jeff Erickson, and Leonidas~J. Guibas.
\newblock Kinetic binary space partitions for intersecting segments and
  disjoint triangles.
\newblock In {\em Proceedings of the 9th Annual {ACM-SIAM} Symposium on
  Discrete Algorithms (SODA)}, pages 107--116, 1998.
\newblock URL: \url{http://dl.acm.org/citation.cfm?id=314613}.

\bibitem{ArgeBG06}
Lars Arge, Gerth~St{\o}lting Brodal, and Loukas Georgiadis.
\newblock Improved dynamic planar point location.
\newblock In {\em Proceedings of the 47th Annual {IEEE} Symposium on
  Foundations of Computer Science {(FOCS})}, pages 305--314, 2006.
\newblock \href {http://dx.doi.org/10.1109/FOCS.2006.40}
  {\path{doi:10.1109/FOCS.2006.40}}.

\bibitem{Bahrdt17}
Daniel Bahrdt and Martin~P. Seybold.
\newblock Rational points on the unit sphere: Approximation complexity and
  practical constructions.
\newblock In {\em Proceedings of the 2017 {ACM} on International Symposium on
  Symbolic and Algebraic Computation (ISSAC)}, pages 29--36, 2017.
\newblock \href {http://dx.doi.org/10.1145/3087604.3087639}
  {\path{doi:10.1145/3087604.3087639}}.

\bibitem{BaumgartenJM94}
Hanna Baumgarten, Hermann Jung, and Kurt Mehlhorn.
\newblock Dynamic point location in general subdivisions.
\newblock {\em Journal of Algorithms}, 17(3):342--380, 1994.
\newblock \href {http://dx.doi.org/10.1006/jagm.1994.1040}
  {\path{doi:10.1006/jagm.1994.1040}}.

\bibitem{BentleyW80}
Jon~Louis Bentley and Derick Wood.
\newblock An optimal worst case algorithm for reporting intersections of
  rectangles.
\newblock {\em {IEEE} Transactions on Computers}, 29(7):571--577, 1980.
\newblock \href {http://dx.doi.org/10.1109/TC.1980.1675628}
  {\path{doi:10.1109/TC.1980.1675628}}.

\bibitem{ChanN15}
Timothy~M. Chan and Yakov Nekrich.
\newblock Towards an optimal method for dynamic planar point location.
\newblock In {\em Proceedings of the 56th Annual {IEEE} Symposium on
  Foundations of Computer Science ({FOCS})}, pages 390--409, 2015.
\newblock \href {http://dx.doi.org/10.1109/FOCS.2015.31}
  {\path{doi:10.1109/FOCS.2015.31}}.

\bibitem{ChazelleI84}
Bernard Chazelle and Janet Incerpi.
\newblock Computing the connected components of d-ranges.
\newblock {\em Bulletin of the {EATCS}}, 22:9--10, 1984.
\newblock URL:
  \url{https://www.cs.princeton.edu/~chazelle/pubs/ComputConnectedComponDranges.pdf}.

\bibitem{Tamassia92survey}
Y.~{Chiang} and R.~{Tamassia}.
\newblock Dynamic algorithms in computational geometry.
\newblock {\em Proceedings of the IEEE}, 80(9):1412--1434, 1992.
\newblock \href {http://dx.doi.org/10.1109/5.163409}
  {\path{doi:10.1109/5.163409}}.

\bibitem{ChiangT92}
Yi{-}Jen Chiang and Roberto Tamassia.
\newblock Dynamization of the trapezoid method for planar point location in
  monotone subdivisions.
\newblock {\em International Journal of Computational Geometry \&
  Applications}, 2(3):311--333, 1992.
\newblock \href {http://dx.doi.org/10.1142/S0218195992000184}
  {\path{doi:10.1142/S0218195992000184}}.

\bibitem{BergCompGeo}
Mark de~Berg, Otfried Cheong, Marc~J. van Kreveld, and Mark~H. Overmars.
\newblock {\em Computational Geometry: Algorithms and Applications, 3rd
  Edition}.
\newblock Springer, 2008.
\newblock \href {http://dx.doi.org/10.1007/978-3-540-77974-2}
  {\path{doi:10.1007/978-3-540-77974-2}}.

\bibitem{EdelsbrunnerO85}
Herbert Edelsbrunner and Mark~H. Overmars.
\newblock Batched dynamic solutions to decomposable searching problems.
\newblock {\em Journal of Algorithms}, 6(4):515--542, 1985.
\newblock \href {http://dx.doi.org/10.1016/0196-6774(85)90030-6}
  {\path{doi:10.1016/0196-6774(85)90030-6}}.

\bibitem{GoodrichT98}
Michael~T. Goodrich and Roberto Tamassia.
\newblock Dynamic trees and dynamic point location.
\newblock {\em {SIAM} Journal on Computing}, 28(2):612--636, 1998.
\newblock \href {http://dx.doi.org/10.1137/S0097539793254376}
  {\path{doi:10.1137/S0097539793254376}}.

\bibitem{Hemmer16}
Michael Hemmer, Michal Kleinbort, and Dan Halperin.
\newblock Optimal randomized incremental construction for guaranteed
  logarithmic planar point location.
\newblock {\em Computational Geometry: Theory and Applications}, 58:110--123,
  2016.
\newblock \href {http://dx.doi.org/10.1016/j.comgeo.2016.07.006}
  {\path{doi:10.1016/j.comgeo.2016.07.006}}.

\bibitem{Hobe18}
Alex Hob{\'e}, Daniel Vogler, Martin~P. Seybold, Anozie Ebigbo, Randolph~R.
  Settgast, and Martin~O. Saar.
\newblock Estimating fluid flow rates through fracture networks using
  combinatorial optimization.
\newblock {\em Advances in Water Resources}, 122:85--97, 2018.
\newblock \href {http://dx.doi.org/10.1016/j.advwatres.2018.10.002}
  {\path{doi:10.1016/j.advwatres.2018.10.002}}.

\bibitem{Mulmuley90}
Ketan Mulmuley.
\newblock A fast planar partition algorithm, {I}.
\newblock {\em Journal of Symbolic Computation}, 10(3-4):253--280, 1990.
\newblock \href {http://dx.doi.org/10.1016/S0747-7171(08)80064-8}
  {\path{doi:10.1016/S0747-7171(08)80064-8}}.

\bibitem{MulmuleyBook}
Ketan Mulmuley.
\newblock {\em Computational Geometry: An Introduction Through Randomized
  Algorithms}.
\newblock Prentice Hall, 1994.

\bibitem{MunroN19}
J.~Ian Munro and Yakov Nekrich.
\newblock Dynamic planar point location in external memory.
\newblock In {\em Proceedings of the 35th International Symposium on
  Computational Geometry (SoCG)}, pages 52:1--52:15, 2019.
\newblock \href {http://dx.doi.org/10.4230/LIPIcs.SoCG.2019.52}
  {\path{doi:10.4230/LIPIcs.SoCG.2019.52}}.

\bibitem{Schwarzkopf91}
Otfried Schwarzkopf.
\newblock Dynamic maintenance of geometric structures made easy.
\newblock In {\em Proceedings of the 32nd Annual IEEE Symposium on Foundations
  of Computer Science (FOCS)}, pages 197--206, 1991.
\newblock \href {http://dx.doi.org/10.1109/SFCS.1991.185369}
  {\path{doi:10.1109/SFCS.1991.185369}}.

\bibitem{Seidel91}
Raimund Seidel.
\newblock A simple and fast incremental randomized algorithm for computing
  trapezoidal decompositions and for triangulating polygons.
\newblock {\em Computational Geometry: Theory and Applications}, 1:51--64,
  1991.
\newblock \href {http://dx.doi.org/10.1016/0925-7721(91)90012-4}
  {\path{doi:10.1016/0925-7721(91)90012-4}}.

\bibitem{SeidelA96}
Raimund Seidel and Cecilia~R. Aragon.
\newblock Randomized search trees.
\newblock {\em Algorithmica}, 16(4-5):464--497, 1996.
\newblock \href {http://dx.doi.org/10.1007/BF01940876}
  {\path{doi:10.1007/BF01940876}}.

\bibitem{Seybold17}
Martin~P. Seybold.
\newblock Robust map matching for heterogeneous data via dominance
  decompositions.
\newblock In {\em Proceedings of the 2017 {SIAM} International Conference on
  Data Mining (ICDM)}, pages 813--821, 2017.
\newblock \href {http://dx.doi.org/10.1137/1.9781611974973.91}
  {\path{doi:10.1137/1.9781611974973.91}}.

\bibitem{CGAL}
{The CGAL Project}.
\newblock {\em {CGAL} User and Reference Manual}.
\newblock {CGAL Editorial Board}, {4.14.1} edition, 2019.
\newblock URL: \url{https://doc.cgal.org/4.14.1/Manual/packages.html}.

\end{thebibliography}

\clearpage
\appendix

\section{Pseudocode for Recursive Primitives} \label{sec:code}

\includegraphics[width=\linewidth]{./figs/vpart-vmerge.pdf}

\begin{algorithm}[h!]
	\caption{$\vPart(u, q, v^-,v^+)$:\hfill Assertion: $c(q)$ intersects $\Delta(u)$.}
	\label{alg:v-part}
	\begin{itemize} 
		\item[] Stop if $u$ is a leaf; Let $(u_l,u_a,u_b,u_r) := \subtrees(u)$.
		\item[\em TST:] {\bf IF} $c(q)$ intersects $\Delta(u_a)$ and $\Delta(u_b)$
		\item[\em TSD:] {\bf IF} $q \in \Delta(u_a)$ \hfill `$q \in \Delta(u_b)$' analogue
		\begin{enumerate}
			\item[]
			If present, move $u_l$ and its parent's vertical cut in $v^-$ \\
			If present, move $u_r$ and its parent's vertical cut in $v^+$
			\item[] Cut both unoccupied leafs, that is one child of $v^-$ and one of $v^+$, with the edge cut. Let $v^-_a,v^-_b,v^+_a,v^+_b$ denote these leafs.
			\item[] 		$\vPart(u_a,q,v^-_a,v^+_a)$
			\item[\em TST:] $\vPart(u_b,q,v^-_b,v^+_b)$
			\item[\em TSD:] Set both below pointers on $u_b$ instead of $v^-_b$ and $v^+_b$.
		\end{enumerate}
		\item[] {\bf ELSE IF} $c(q)$ intersects $\Delta(u_l)$ \hfill `$c(q)$ intersects $\Delta(u_r)$' symmetrically
		\begin{enumerate}
			\item[] Move $u_r$ (if present) $,u_a,u_b$ and their parents' cuts in $v^+$.
			\item[] Let $v^+_l$ denote the unoccupied leaf.
			\item[] $\vPart(u_l,q,v^-,v^+_l)$ 
		\end{enumerate}
	\end{itemize} 
\end{algorithm}

\begin{algorithm}[h!]
	\caption{$\vMerge(u^-, u^+, q, v)$: \hfill Assertion: $c(q)$ bounds $\Delta(u^-)$ and $\Delta(u^+)$}
	\label{alg:v-merge}
	\begin{itemize} 
		\item[] {\bf IF} $u^-$ is leaf {\bf THEN} Move contents of $u^+$ in $v$ and stop. \hfill `$u^+$ leaf' analogue
		\item[] Let $(u^-_l,u^-_a,u^-_b,u^-_r) := \subtrees(u^-)$ and $(u^+_l,u^+_a,u^+_b,u^+_r) := \subtrees(u^+)$
		\item[\em TST:] {\bf IF} $p(u^-)=p(u^+)$ 
		\item[\em TSD:] {\bf IF} $p(u^-)=p(u^+)$ and $u^-_b=u^+_b$\hfill `$u^-_a=u^+_a$' analogue
		\begin{enumerate}
			\item[] If present, move $u^-_l$ and its parent's vertical cut in $v$
			\item[] If present, move $u^+_r$ and its parent's vertical cut in the unoccupied leaf of $v$
			\item[] Cut the new unoccupied leaf of $v$ with the edge cut. Let $v_a,v_b$ denote these leafs.
			\item[] 	    $\vMerge(u^-_a,u^+_a,q,v_a)$ 
			\item[\em TST:] $\vMerge(u^-_b,u^+_b,q,v_b)$
			\item[\em TSD:] Set the below pointer on $u^-_b$ instead of $v_b$.
			
		\end{enumerate}
		\item[] {\bf ELSE IF} $p(u^-)<p(u^+)$ \hfill `$>$' symmetrically
		\begin{enumerate}
			\item[] Move $u^-_l$ (if present) $, u^-_a, u^-_b$ and their parents' cuts in $v$.
			\item[] Let $v_r$ denote the unoccupied leaf.
			\item[] $\vMerge(u^-_r,u^+,q,v_r)$ 
		\end{enumerate}
	\end{itemize} 
\end{algorithm}

\begin{algorithm}[h!]
	\caption{$\partition(u,c,v^-,v^+)$:\hfill Assertion: $c$ crosses $\Delta(u)$ entirely}
	\label{alg:partition}
	
	\includegraphics[width=\linewidth]{./figs/partition.pdf}
\\
	\begin{enumerate}
		
		\item[] Stop if $u$ is a leaf
		\item[] Let $(u_l,u_a,u_b,u_r) := \subtrees(u)$ and $l_1,l_2$ denote the vertical cuts destroying $\Delta(u)$
		
		\item[]{\bf IF} $c$ does not intersect $\Delta(u_b)$ properly \hfill `not $\Delta(u_a)$' analogue
		\begin{itemize}
			\item[] Allocate new nodes $v^-_l,v_l^+,v_a^-,v_a^+, v_r^-, v_r^+, v^+_{ar}$.
			
			\item[] $\partition(u_l, c, v_l^-,v_l^+)$; $\partition(u_a, c,v^-_a,v^+_a)$; $\partition(u_r, c,v^-_r,v^+_r)$
			
			\item[] $\vMerge(v^+_a, v^+_r, l_2, v^+_{ar})$\\
			$\vMerge(v^+_l, v^+_{ar}, l_1, v^+)$
			\item[] Place $v^-_l, v^-_a, u_b, v^-_r$ below respective cuts under $v^-$
		\end{itemize}
		
		\item[] {\bf ELSE IF} edge cut of $u$ is steeper than $c$ \hfill `less steep' symmetrically
		\begin{itemize}
			\item[] Let $i$ denote the vertical cut (induced by the intersection of the edge-cut and $c$)
			\item[] Allocate new nodes $v^-_l,v_l^+, v_r^-, v_r^+$ as well as  $v_{ar},v_{al},v_{al}^-, v_{al}^+$ and
			$v_{bl},v_{br},v_{br}^-, v_{br}^+$.

			\item[] $\partition(u_l, c, v^+_l, v^-_l)$;
					$\partition(u_r, c, v^+_r,v^-_r)$

			\item[] $\vPart(u_a,i, v_{al},v_{ar})$;
			        $\vPart(u_b,i, v_{bl},v_{br})$
						
			\item[] $\partition(v_{al}, c, v^-_{al},v^+_{al})$; 
					$\partition(v_{br}, c, v^-_{br},v^+_{br})$\\
			
			\item[] Place $\vMerge(v^+_l, v^+_{al}, l_1, \cdot )$ as left child under cut $i$ below $v^+$\\
			Place $v^+_r$ and its parent's cut in the unoccupied leaf of $v^+$\\
			Cut the unoccupied leaf of $v^+$ with the edge cut of $u$\\
			Place $v_{ar}$ and $v^+_{br}$ in these leafs\\
			
			\item[] Place $v^-_l$ and its parent's  cut in $v^-$\\
			Place $\vMerge(b^-_{br}, v^-_r, l_2, \cdot)$ as right child under cut $i$ in $v^-$\\
			Cut the unoccupied leaf of $v^-$ with the edge cut of $u$\\
			Place $v^-_{al}$ and $v_{bl}$ in these leafs
			
		\end{itemize}

	\end{enumerate}
\end{algorithm}

\begin{algorithm}[h!]
	\includegraphics[width=\linewidth]{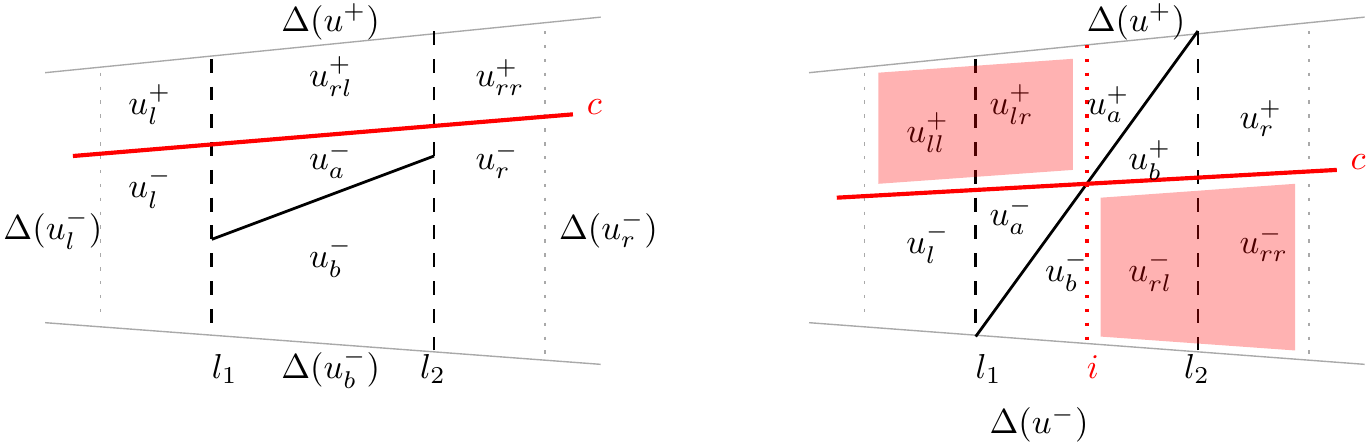}

	\caption{$\merge(u^-,u^+,c,v)$:\hfill Assertion: $c$ bounds $\Delta(u^-)$ and $\Delta(u^+)$}
	\label{alg:merge}

	\begin{enumerate} 
		\item[] {\bf IF} $u^-$ is leaf {\bf THEN} Move contents of $u^+$ in $v$ and stop. \hfill `$u^+$ leaf' analogue

		\item[]{\bf IF} $p(u^-) < p(u^+)$  \hfill `$>$' symmetrically
	\begin{enumerate}
		\item[] Let $(u^-_l,u^-_a,u^-_b,u^-_r) := \subtrees(u^-)$ and $l_1,l_2$ denote the vertical cuts of $u^-$
		\item[] Move the cuts of $u^-$ in $v$ and place $u^-_b$ below the edge-cut;  
		Let $v_l,v_a,v_r$ denote the unoccupied leafs
		\item[] $\vPart(u^+,l_1,u^+_l,u^+_r)$
		\item[] $\vPart(u^+_r,l_2,u^+_{rl},u^+_{rr})$
		\item[] $\merge(u^-_l, u^+_l,c,v_l)$; $\merge(u^-_r, u^+_{rr},c,v_r)$; $\merge(u^-_a, u^+_{rl},c,v_a)$
	\end{enumerate}
		\item[]{\bf ELSE IF} $p(u^-)=p(u^+)$ and edge cut is steeper than $c$ \hfill`less steep' symmetrically
		\begin{enumerate}

		\item[] Let $(u^-_l,u^-_a,u^-_b,u^-_r) := \subtrees(u^-)$ and $(u^+_l,u^+_a,u^+_b,u^+_r) := \subtrees(u^+)$
		\item[] Let $l_1$ be the vertical cut of $u^-$, $l_2$ of $u^+$, and $i$ their intersection cut
		\item[] Cut $v$ with $l_1,l_2$ and the (common) edge cut of $u^-$ and $u^+$. Let $v_l,v_a,v_b,v_r$ denote the unoccupied leafs
		
		\item[] $\vPart(u^+_l,l_1,u^+_{ll},u^+_{lr})$
		\item[] $\vPart(u^-_r,l_2,u^-_{rl},u^-_{rr})$
		
		\item[] $\merge(u^-_l,u^+_{ll}, c, v_l)$; $\merge(u^-_{rr},u^+_r, c, v_r)$\\
		
		\item[] $\merge(u^+_b,u^-_{rl},c,v_{br})$; $\merge(u^-_a,u^+_{lr},c,v_{al})$
		\item[] $\vMerge(v_{al},u^+_a,i,v_a)$; $\vMerge(u^-_b,v_{br},i,v_b)$;
\end{enumerate}
	\end{enumerate}
\end{algorithm}

\end{document}